\RequirePackage{fix-cm}
\documentclass[smallcondensed]{svjour3}     
\smartqed  
\usepackage{graphicx}
\usepackage{latexsym,amsmath,enumerate,amssymb,amsbsy}
\usepackage{supertabular,longtable,multirow}
\usepackage{graphics,color}
\usepackage{verbatim}
\usepackage{url,breakurl}
\usepackage{dblfloatfix}
\usepackage{mathptmx}      
\newcommand{\F}{\mathbb F}
\newcommand{\C}{\mathbb C}
\newcommand{\D}{\mathcal{D}}
\newcommand{\cE}{\mathcal{E}}
\newcommand{\cF}{\mathcal{F}}
\newcommand{\W}{\mathcal{W}}
\newcommand{\V}{\mathcal{V}}
\newcommand{\M}{\mathcal{M}}
\newcommand{\cN}{\mathcal{N}}
\newcommand{\cC}{\mathcal{C}}

\newcommand{\bv}{\boldsymbol{v}}
\newcommand{\bw}{\boldsymbol{w}}
\newcommand{\0}{\boldsymbol{0}}
\newcommand{\E}{\mathop{{\rm E}}}
\newcommand{\Tr}{\mathop{{\rm Tr}}}
\newcommand{\N}{\mathop{{\rm N}}}
\newcommand{\flr}[1]{\left \lfloor #1 \right \rfloor}

\def\dim{\mathop{{\rm dim}}}

\def\min{\mathop{{\rm min}}}
\def\u{{\mathbf{u}}}

\def\v{{\mathbf{v}}}
\def\wt{\mathop{{\rm wt}}}
\def\1{{\mathbf{1}}}

\journalname{Designs, Codes and Cryptography}
\begin{document}

\title{Xing-Ling Codes, Duals of their Subcodes,\\and Good Asymmetric Quantum Codes
\thanks{The work of S.~Jitman was partially supported by the National Research Foundation of Singapore under Research Grant NRF-CRP2-2007-03.\\ 
The Centre for Quantum Technologies is a Research Centre of Excellence funded by the Ministry of Education and the National Research Foundation of Singapore.\\
The authors' collaboration leading to this work was facilitated by travel grants provided by the Merlion Project No. 1.02.10.}
}
\author{Martianus Frederic Ezerman \and Somphong Jitman \and Patrick Sol\'{e}}
\authorrunning{Ezerman \and Jitman \and Sol\'{e}} 

\institute{Martianus Frederic Ezerman \at 
	Centre for Quantum Technologies (CQT), National University of Singapore,\\
	Block S15, 3 Science Drive 2, Singapore 117543.\\
    \email{frederic.ezerman@gmail.com, cqtmfe@nus.edu.sg}\\
           \and
        Somphong Jitman \at
	Department of Mathematics, Faculty of Science, Silpakorn University, Nakhonpathom 73000, Thailand.\\
	\emph{Former address: Division of Mathematical Sciences, School of Physical and Mathematical Sciences,\\
	Nanyang Technological University, 21 Nanyang Link, Singapore 637371.}\\
	\email{somphong@su.ac.th}\\
	    \and
	Patrick Sol\'{e} \at
	Telecom ParisTech, 46 rue Barrault, 75634 Paris Cedex 13, France, and\\
	Mathematics Department, King Abdulaziz University, Jeddah, Saudi Arabia.\\
	\email{sole@enst.fr}
}

\date{Received: date / Accepted: date}

\maketitle

\begin{abstract}
	A class of powerful $q$-ary linear polynomial codes originally proposed by Xing and Ling is deployed to construct good 
	asymmetric quantum codes via the standard CSS construction. Our quantum codes are $q$-ary block codes that encode $k$ qudits 
	of quantum information into $n$ qudits and correct up to $\flr{(d_{x}-1)/2}$ bit-flip errors and up to
	$\flr{(d_{z}-1)/2}$ phase-flip errors.. In many cases where the length $(q^{2}-q)/2 \leq n \leq (q^{2}+q)/2$ 
	and the field size $q$ are fixed and for chosen values of $d_{x} \in \{2,3,4,5\}$ and $d_{z} \ge \delta$, where $\delta$ 
	is the designed distance of the Xing-Ling (XL) codes, the derived pure $q$-ary asymmetric quantum CSS codes possess the best 
	possible size given the current state of the art knowledge on the best classical linear block codes.
\end{abstract}

\section{Introduction}\label{intro}
The ability to perform quantum error-correction is essential in many quantum information processing tasks. In various scenarios involving 
qubit channels one can benefit from the presence of asymmetry in the respective probabilities of the bit-flip and the phase-flip errors. 

In the combined amplitude damping and dephasing channel investigated in~\cite{IM07} and~\cite{SKR09}, for example, the probabilities of 
the bit and the phase flips are related to the relaxation and the dephasing time, respectively, whose ratio can then be used to quantify 
the channel's asymmetry.

A scheme for fault-tolerant quantum computation that works effectively against highly biased noise, where dephasing 
is far stronger than all other types of noise, is presented in~\cite{AP08}. The accuracy threshold for quantum computation are shown to be 
improved by exploiting this noise asymmetry.

Quantum codes tailored to handle a particular ratio of asymmetry present in the channel are usually called asymmetric quantum codes 
(AQCs). Of the systematic construction methods for AQCs, the most widely used is the standard CSS construction, named after 
Calderbank, Shor, and Steane, that links a pair of nested classical linear $q$-ary codes to a $q$-ary AQC whose parameters can 
be directly deduced from the parameters of the corresponding classical code pair. More background materials on the theoretical 
derivations of the basic facts can be found in~\cite{WFLX09} and the references cited therein. 

Let $C_{i}$ for $i \in \{1,2\}$ be a $q$-ary linear code with length $n$, dimension $k_{i}$ and minimum distance $d_{i}$. 
Let $C_{i}^{\perp}$ be the Euclidean dual of $C_{i}$. 
To design a standard CSS AQC with good parameters the following three requirements need to be satisfied. First, we require that 
$C_{1}^{\perp} \subset C_{2}$. Second, for fixed values of $\left(q,n,d_{2}\right)$, the dimension $k_{2}$ must be as large as possible. 
Third, $k_{1}$ must also be as large as possible for specified values of $\left(q,n,d_{1}\right)$. Since the Euclidean inner product 
is non-degenerate, the third condition implies that the codimension of $C_{1}^{\perp}$ in $C_{2}$ should be as large as possible.

Since the requirements are well-understood, many families of nested classical codes have been recognized as natural choices in the construction. 
Prior works have made use of cyclic codes and their subfamilies such as the BCH and the quadratic residue (QR) codes. Other families that have 
been investigated include the low-density parity-check (LDPC) codes, the Reed-Muller codes, the Reed-Solomon codes and their generalization, 
the character codes, the affine-invariant and the product codes. Classical propagation methods have also been applied to construct AQCs of higher 
lengths based on already constructed ones. The tables provided in~\cite{Gua13} and the references therein provide a summary of previously 
constructed families of AQCs. Note that most of the results in the MDS family presented in~\cite[Table 2]{Gua13} were mistakenly attributed 
to~\cite{WFLX09} instead of to the correct source~\cite{EJKL11}.

This paper derives $q$-ary CSS AQCs with good parameters based on nested polynomial codes first introduced by Xing and Ling in~\cite{XL2} and, 
henceforth, called XL codes. These codes fit nicely into the standard CSS framework since they generally have good parameters for their 
range of lengths $(q^{2}-q)/2 \leq n \leq (q^{2}+2)/2$ and their nestedness is self-evident. However, since the dual of an XL code is not 
necessarily an XL code, the challenge is to understand the structure of the duals of some carefully chosen subcodes of large codimension 
of the XL codes. This, to the best of our knowledge, had not been considered before.

Our investigation leads to a more complete picture in the studies of AQCs based on XL codes. For $d_{x} \in \{2,3,4\}$ the exact parameters 
of the resulting AQCs are explicitly determined. For $d_{x}=5$ we have enough information to set a good lower bound on the distances while all 
other parameters can be easily derived from the properties of the classical pair. For prime power $q \leq 9$, the data on currently best-known $q$-ary 
linear block codes in Grassl's online tables~\cite{Gr09} can be efficiently used to provide a performance benchmark as a measure of optimality.

After this introduction, Section~\ref{sec:prelims} reviews important definitions and establishes a certificate of optimality based on the state 
of the art knowledge on classical codes. A summary of the construction and parameters of the XL codes is provided in Section~\ref{sec:XL} where 
particular attention is given to subcodes of XL codes with dual distances in the set $\{2,3,4,5\}$. The parameters of the resulting pure $q$-ary AQCs are 
then computed explicitly for $q \in \{3,4,5,7,8,9\}$ in Section~\ref{sec:AQC}. In many instances they can be certified to be optimal or best-known 
based on Theorem~\ref{th:opt} given in Section~\ref{sec:prelims}. A summary and several open directions form the last section.

All computations in this work are done using MAGMA~\cite{BCP97} Version 2.19-3.

\section{Preliminaries}\label{sec:prelims}
Let $q$ be a prime power and ${\F}_{q}$ be the finite field having 
$q$ elements. A {\it linear $[n,k,d]_q$-code} $C$ is a $k$-dimensional ${\F}_{q}$-subspace of ${\F}_{q}^n$ with 
{\em minimum distance} $d:=\min \{ \wt(\v) | \v \in C \setminus \{\0 \}\}$, where $\wt(\v)$ denotes the 
{\it Hamming weight}  of $\v \in {\F}_{q}^{n}$. Given two distinct linear codes $C$ and $D$, $\wt(C \setminus D)$ denotes 
$\min \{ \wt(\u) | \u\in C \setminus D \}$. Given $q,n$, and $d$, let $B_{q}(n,d)$ denote 
$\max \{q^{k} | \text{ there exists an } [n,k,d]_{q}\text{-code}\}$.

When discussing a specific code $C$, we use $d(C)$ and $\dim(C)$ to denote its minimum distance and dimension as an $\F_{q}$-subspace, respectively.

For $\u=(u_i)_{i=1}^n$, $\v=(v_i)_{i=1}^n \in \F_{q}^{n}$, their \textit{Euclidean inner product} 
is given by $(\u,\v)_{\E}:=\sum_{i=1}^{n} u_{i} \cdot v_{i}$. With respect to this inner product, the {\it dual} $C^{\perp}$ of $C$ is given by 
\begin{equation*}
C^\perp:=\left\lbrace \u \in {\F}_q^n | (\u,\v )_{\E} = 0 \text{ for all } \v \in C \right\rbrace.
\end{equation*}

Let $d_{x}$ and $d_{z}$ be positive integers. A quantum code $Q$ in 
$V_{n}=({\C}^{q})^{\otimes n}$ with dimension $K \geq 1$ is called an 
\textit{asymmetric quantum code} with parameters 
$((n,K,\{ d_{z} , d_{x} \}))_{q}$, or $[[n,k,\{ d_{z} , d_{x} \}]]_{q}$ with $k=\log_{q} K$ whenever $Q$ is a stabilizer code, 
if $Q$ is able to detect any combination of up to $d_{x}-1$ bit-flips (or $X$-errors) and up to $d_{z}-1$ phase-flips 
(or $Z$-errors) simultaneously. 
 
The standard CSS construction is given in {\it e.g.}~\cite{AA10,WFLX09}.
\begin{theorem}\label{css}
Let $C_{i}$ be linear codes with parameters $[n,k_{i},d_{i}]_{q}$ for $i \in \{1,2\}$ 
with $C_{1}^{\perp}\subseteq C_{2}$. Let 
\begin{equation}\label{eq:distances}
d_{z}:=\wt(C_{2} \setminus C_{1}^{\perp}) \text{ and } d_{x}:=\wt(C_{1} \setminus C_{2}^{\perp}) \text{.}
\end{equation}
Then there exists an AQC $Q$ with parameters $[[n,k_{1}+k_{2}-n,\{ d_{z} , d_{x} \}]]_{q}$. The code $Q$ is said to be 
\textit{pure} whenever $d_{z}=d_{2}$ and $d_{x}=d_{1}$.
\end{theorem}

\begin{remark}\label{rem:switch}
All CSS codes are stabilizer codes. In the literature it is customary to assume $d_z \ge d_x$ since in general 
the dephasing errors occur with higher probability than the bit-flip errors do. In this paper we opt not to order 
the distances to better present how their computational values are derived. Whenever necessary, one can apply a Fourier 
transformation over $\F_q$ to interchange the role of the bit-flip and the phase-flip error operators. 
That way, $d_z \ge d_x$ can be obtained.
\end{remark}

The purity in Theorem~\ref{css} is equivalent to the general definition given in~\cite[Th. 3.1 Part (ii)]{WFLX09}. 
A certificate of optimality for pure $q$-ary CSS AQCs can be based on the following result.

\begin{theorem}\label{th:opt}
If there exist a pure standard CSS $[[n,k,\{ d_{z} , d_{x} \}]]_q$ code $Q$, then 
\begin{equation}\label{bound}
k\leq \log_{q}(B_q(n,d_x))+\log_{q}(B_q(n,d_z))-n.
\end{equation}
\end{theorem}
\begin{proof}
Assume there exists a  pure CSS $[[n,k,\{ d_{z} , d_{x} \}]]_q$ code. Then, equivalently, there exist $q$-ary linear 
codes $C_1$ and $C_2$ such that $d(C_1)=d_x$,  $d(C_2)=d_z$, $C_1^\perp \subset C_2$ and $k=\dim(C_2)-\dim(C_1^\perp)$. Since 
$\dim(C_1)\leq \log_{q}(B_q(n,d_x))$ and $\dim(C_2)\leq \log_{q}(B_q(n,d_z))$, 
\begin{align*}
	k&=\dim(C_2)-\dim(C_1^\perp) = \dim(C_2)-(n-\dim(C_1))\\
	 & \leq \log_{q}(B_q(n,d_x))+\log_{q}(B_q(n,d_z))-n.
\end{align*}
\end{proof}

The bound~(\ref{bound}) holds true for pure CSS AQCs. Fixing $n$ and $d_{i}$ for $i \in \{1,2\}$, if both $k_{1}$ and $k_{2}$ 
are optimal, then the cardinality of $Q$ is optimal among CSS AQCs of equal parameter set $(q,n,d_{z},d_{x})$. If both $C_1$ 
and $C_2$ have the same dimension as the currently best-known linear codes listed in~\cite{Gr09}, then $Q$ has the currently 
best-known cardinality among comparable CSS AQCs. Any improvement on the lower bound of $B_{q}(n,d_{i})$ potentially leads to an 
improved quantum code and there would not be any improvement on the parameters of an AQC if there are no improvements on the 
lower bound of the corresponding $B_{q} (n, d_{i})$.

To end this section, we recall two mappings from $\F_{q^{2}}$ onto $\F_{q}$ which will be used extensively in what follows. 
The trace mapping $\Tr$ sends $\gamma$ to $\gamma + \gamma^{q}$ while the norm mapping $\N$ outputs $\gamma^{q+1}$ on input $\gamma$. 
Properties and important results concerning these two mappings in the more general setup of $\F_{q^{m}}$ for positive integer $m$ are 
discussed in details in~\cite[Ch. 2 Sect. 3]{LN97}.

\section{Suitably Chosen Nested XL Codes}\label{sec:XL}
This section is presented in two parts. In the first subsection we recall the construction of XL codes and their parameters. 
In the second subsection, we construct nested XL codes of the right parameters to use in the CSS construction.

\subsection{Construction of XL Codes}\label{subsec:construct}
In this subsection, we recall the construction of the XL codes given in~\cite{XL2}.
For a finite field $\F_{q}$, let $\F_{q^2}$ be its quadratic extension. Let $\{\alpha_1,\alpha_2,\dots,\alpha_q\}$ 
be a fixed list of the elements in $\F_{q}$. Without loss of generality, let us assume that $\F_{q^2}$ is listed as 
\begin{equation}\label{eq:choice}
\{\alpha_{1},\alpha_{2},\dots,\alpha_{q},\beta_{1},\beta_{1}^{q},\beta_{2},\beta_{2}^{q},\dots,\beta_{r},\beta_{r}^{q}\} \text{, where }
r=(q^2-q)/2\text{.}
\end{equation}

Define $V_{1,0}$ to be the $\F_{q}$-vector space generated by the polynomial~$1$. For $2\leq m\leq q-1$ and $0\leq \ell \leq m-1$, 
let
\begin{equation*}
V_{m,\ell}:=\langle \{e_{i,j}(x) | 0\leq i\leq j\leq m-2\} \cup \{e_{i,m-1}(x)| 0\leq i\leq \ell\}\rangle \text{,}
\end{equation*}
where
\begin{equation}
e_{i,j}(x)=
\begin{cases}\label{basis}
x^{iq+j}+x^{jq+i} & \text{ if } i\ne j \text{,}\\
x^{iq+j}& \text{ if } i= j \text{,}
\end{cases}
\end{equation}
for all $i,j\geq 0$.

For easy reference, we explicitly list $e_{i,j}(x)$ for $0 \leq i \leq j \leq 4$ down in Figure 1.
\begin{figure*}[!h]
\normalsize
\centering
\begin{tabular}{c c c c c c c}
    \hline\noalign{\smallskip}
    \multicolumn{2}{c}{\multirow{2}{*}{$e_{i,j}(x)$}}&\multicolumn{5}{c}{$j$}\\
	\multicolumn{2}{c}{} & $0$ & $1$ & $2$ & $3$ & $4$\\
	\noalign{\smallskip}
	\hline\noalign{\smallskip}
    \multirow{5}{*}{$i$} & $0$ & $1$ & $x^{q}+x$ & $x^{2q} + x^{2}$     & $x^{3q} + x^{3}$      & $x^{4q} + x^{4}$ \\
	\noalign{\smallskip}
                         & $1$ &     & $x^{q+1}$ & $x^{2q+1} + x^{q+2}$ & $x^{3q+1} + x^{q+3}$  & $x^{4q+1} + x^{q+4}$ \\
	\noalign{\smallskip}
                         & $2$ &     &           & $x^{2q+2}$           & $x^{3q+2} + x^{2q+3}$ & $x^{4q+2} + x^{2q+4}$ \\
	\noalign{\smallskip}
                         & $3$ &     &           &                      & $x^{3q+3}$            & $x^{4q+3} + x^{3q+4}$ \\
	\noalign{\smallskip}
                         & $4$ &     &           &                      &                       & $x^{4q+4}$ \\
    \noalign{\smallskip}\hline
	\noalign{\smallskip}
    \multicolumn{7}{c}{\multirow{2}{*}{Figure 1: List of $e_{i,j}(x)$ for $0 \leq i \leq j \leq 4$}}
\end{tabular}
\label{fig:one}
\end{figure*}

Given that $0\leq t\leq q$, $2\leq m\leq q-1$ and $0\leq \ell \leq m-1$, the $q$-ary linear code $C_q(t,m,\ell)$
is defined as the evaluation code of $\{f(x) \in V_{m,\ell}\}$ on $(\alpha_{1},\dots,\alpha_{t},\beta_{1},\dots,\beta_{r})$. 
Explicitly,
\begin{equation}\label{eq:CodeC}
C_q(t,m,\ell):= \{(f(\alpha_{1}),\dots,f(\alpha_{t}),f(\beta_{1}),\dots,f(\beta_{r}))\mid f(x)\in V_{m,\ell}\}\text{.}
\end{equation}
When $t$ is set to be $0$, none of the elements of $\F_{q}$ is chosen in the evaluation. For $q \leq 9$, Table~\ref{table:albe} lists down our 
choices of $\alpha_{i}$ for $1 \le i \le q$ and $\beta_{j}$ for $1 \le j \le r$.
\begin{table}[!hbt]
\caption{Actual Choices for $\alpha_{i}$ and $\beta_{j}$ used in Computation}
\label{table:albe}
\centering
\begin{tabular}{c l l l }
		\hline\noalign{\smallskip}
		$q$ &   $ a \in \F_{q^2}$ root of & $\F_{q}=\{\alpha_1,\dots,\alpha_q\}$ & $\{\beta_1,\beta_2,\dots,\beta_r\} \subset \F_{q^{2}}$ \\
		\noalign{\smallskip}\hline\noalign{\smallskip}
		$3$ & $x^2 + 2x + 2$ & $ \{1, 2, 0 \}$ 			 & $\{a, a^2, a^5\}$ \\
		\noalign{\smallskip}
		$4$ & $x^4 + x + 1$  & $\{1, 0, a^5, a^{10}\}$	 & $\{a, a^2, a^3, a^6, a^7,a^{11} \}$\\
		\noalign{\smallskip}
		$5$ & $x^2 + 4x + 2$ & $\{1, 4, 0, 2, 3\}$		 & $\{a, a^2, a^3, a^4, a^7, a^8, a^9, a^{13}, a^{14}, a^{19}\}$\\
		\noalign{\smallskip}
		$7$ & $x^2 + 6x + 3$ & $\{1, 6, 0, 2, 5, 3, 4\}$ & $\{a, a^2, a^3, a^4, a^5, a^6, a^9, a^{10}, a^{11}, a^{12}, a^{13}, a^{17},$ \\
		    &                &                           & $\enskip a^{18}, a^{19}, a^{20},a^{25}, a^{26}, a^{27}, a^{33}, a^{34}, a^{41}\}$ \\
		\noalign{\smallskip}
		$8$ & $x^6 + x^4 + x^3$ & $\{1, a^{45}, a^{36}, a^{27},$ &
		      $\{a, a^2, a^3, a^4, a^5, a^6, a^7, a^{10}, a^{11}, a^{12}, a^{13}, a^{14}, a^{15}, a^{19}, a^{20},$\\
		    & $+ x + 1$  &$\enskip a^{18}, 0, a^9, a^{54}\}$ & $\enskip a^{21}, a^{22}, a^{23}, a^{28}, a^{29}, a^{30}, a^{31},  a^{37}, a^{38}, a^{39}, a^{46}, a^{47}, a^{55}\}$ \\
		\noalign{\smallskip}
		$9$ & $x^4 + 2x^3 + 2$ & $\{1, 0, a^{70}, a^{60}, a^{50},$
		    & $\{a, a^2, a^3, a^4, a^5, a^6, a^7, a^8, a^{11}, a^{12}, a^{13}, a^{14}, a^{15}, a^{16}, a^{17},$\\
 &&$\enskip 2, a^{30}, a^{20}, a^{10}\}$ & $\enskip  a^{21},a^{22}, a^{23}, a^{24}, a^{25}, a^{26}, a^{31}, a^{32}, a^{33}, a^{34}, a^{35}, a^{41}, a^{42},$\\ 
 &&                                      & $\enskip a^{43}, a^{44},a^{51}, a^{52}, a^{53}, a^{61}, a^{62}, a^{71} \}$\\
		\noalign{\smallskip}\hline
\end{tabular}
\end{table}

\begin{theorem}[{\cite{XL2}}]\label{lx}
Let $0\leq t\leq q$, $2\leq m\leq q-1$ and $0\leq \ell \leq m-1$. Let $h = \binom{m}{2}$, $r=(q^{2}-q)/2$, and 
\[
g=
\begin{cases}
\min\{\max\{2(m-2),m+\ell-1\},t\} & \text{ if } q \text{ is odd,}\\
\max\{\min\{m-2,t\},2t-q\} & \text{ if } q \text{ is even and } \ell \leq m-2 \text{,}\\
\max\{\min\{m-1,t\},2t-q\} & \text{ if } q \text{ is even and } \ell = m-1 \text{.}
\end{cases}
\]
Then $C_q(t,m,\ell)$ as defined in (\ref{eq:CodeC}) is an $[n,k,d]_q$-code 
with $n=t+r$, $k= h+\ell+1$, and 
\begin{equation}\label{ddd}
d\geq \delta := n-\frac{1}{2}(q(m-1)+\ell + g)\text{.}
\end{equation}
\end{theorem}

Note that $V_{m,i}\subset V_{s,j}$ for all $m<s$ or for $i < j$ when $m=s$. Hence, 
$C_q(t,m,i) \subset C_q(t,s,j)$ and
\begin{equation}\label{codim}
\dim(C_q(t,s,j))-\dim(C_q(t,m,i))=\binom{s}{2} -\binom{m}{2} +j - i \text{.}
\end{equation}

\subsection{Suitable Nested XL Code Pairs}\label{sebsec:3.2}
We begin with an easy result to eventually help us construct good AQCs with $d_{x}=2$.
\begin{proposition}\label{1a}
For $0 \leq t \leq q$, The code $C_{q}(t,1,0)^\perp$ is a $[t+r,t+r-1,2]_q$-MDS code.
\end{proposition}
\begin{proof}
This follows immediately since $V_{1,0} = \langle 1 \rangle$ implies that 
$C_q(t,1,0)$ is the $[t+r,1,t+r]_{q}$-repetition code.
\end{proof}

The following is a useful tool in the sequel.
\begin{lemma}\label{a0} Select $s$ distinct elements $a_{1},a_{2},\dots,a_{s} \in 
S:=\{\alpha_1,\alpha_2,\dots,\alpha_t,\beta_1,\beta_2,\dots,\beta_r\}$ as defined in (\ref{eq:choice}). 
If $a_{i}^{q}+a_{i}=a_{j}^{q}+a_{j}$ for all $1 \leq i < j \leq s$, then $a_{i}^{q+1} \neq a_{j}^{q+1}$.
\end{lemma}
\begin{proof}
It suffices to show that for distinct $a,b\in S$, 
\[
a^{q}+a=b^{q}+b \text{ implies } a^{q+1}\neq b^{q+1}\text{.}
\]
For a contradiction, suppose that 
\begin{align}
\Tr(a) = a^{q}+a & = b^{q}+b =\Tr(b) \text{ and}\label{1}\\
\N(a) = a^{q+1} &= b^{q+1} = \N(b) \text{.}\label{2}
\end{align}
Substituting $b^q=a^q+a-b$ from (\ref{1}) into (\ref{2}) yields 
\[
(a^q-b)(a-b)=0 \text{,}
\]
implying $a^{q}=b$ since $a \neq b$.

If $a \in \F_{q}$, then $a=a^q=b$, a contradiction.

If $a \in \{\beta_1,\beta_2,\dots,\beta_r\}$, 
then $b=a^q \notin \{\beta_1,\beta_2,\dots,\beta_r\}\cup \F_{q}$ which contradicts $b \in S$.
\end{proof}
\begin{remark}
Using a result in the theory of finite fields (see~\cite[Exercise 2.24]{LN97}) one can easily infer 
that if there are two elements $a$ and $b$ in $\F_{q^{2}}$ such that $\Tr(a)=\Tr(b)$ 
and $\N(a)=\N(b)$, then they have the same minimal polynomial over $\F_{q}$. Hence, either 
$a=b$ or $a = b^{q}$. This constitutes an alternative proof to Lemma~\ref{a0}.
\end{remark}

Now we are ready to construct a subcode of an XL code with dual distance at least $3$.
\begin{proposition}\label{2a}
Let $q \geq 4$ and $\D=C_{q}(t,2,1)$ be the evaluation code associated with $V_{2,1}$. Then $\D^{\perp}$ is a 
$[t+r,t+r-3,d^{\perp}]_q$-code and $\D \subset C_q(t,m,\ell)$ for all  $m\geq 3$. Moreover, $d^{\perp}=4$ for $(q,t)=(4,0)$, 
while for all other cases $d^{\perp}=3$.
\end{proposition}
\begin{proof}
Since $\D$ is a $[t+r,3,d_{\D}]_{q}$-code, it is clear that $\D^{\perp}$ has length $t+r$ and dimension $t+r-3$. 
By how $\D$ is defined, evaluating based on the given basis $\{1,x^{q}+x,x^{q+1}\}$ gives us a generator matrix 
$G_{\D}=\left(\mathcal{A}|\mathcal{B}\right)$ with
\begin{equation}\label{submat}
\mathcal{A}:=
\left(\begin{array}{c c c c}
1&1&\dots&1\\
\alpha_1^q+\alpha_1	&\alpha_2^q+\alpha_2&\dots&\alpha_t^q+\alpha_t\\
\alpha_1^{q+1} 	&\alpha_2^{q+1} &\dots&\alpha_t^{q+1}
\end{array}\right) \text{, }
\mathcal{B}:=
\left(\begin{array}{c c c c}
1&1&\dots&1\\
\beta_1^q+\beta_1&\beta_2^q+\beta_2&\dots&\beta_r^q+\beta_r\\
\beta_1^{q+1}&\beta_2^{q+1}&\dots&\beta_r^{q+1}
\end{array}\right)\text{.}
\end{equation}
Using Lemma~\ref{a0}, it is easy to verify that any two distinct columns of $G_{\D}$ must be linearly independent. 
Hence, $d^{\perp} \geq 3$. 

Following how the elements of $\F_{q^{2}}$ are defined in Table~\ref{table:albe}, let $\bv(\alpha_{i})$, for $1 \leq i \leq q$, 
be the column vector of $G_{\D}$ associated with the element $\alpha_{i}$. Similarly, let $\bw(\beta_{j})$, for $1 \leq j \leq r$, 
be the column vector associated with the element $\beta_{j}$. The matrix $G_{\D}$ for $t=q$ can then be explicitly constructed 
and erasing the appropriate column(s) from $G_{\D}$ gives us the matrix $G_{\D}$ for lower values of $t$. 

Note that if one chooses a different ordering of the elements in the sets $\{\alpha_{1}\}_{i=1}^{q}$ and 
$\{\beta_{j}\}_{j=1}^{r}$, then the resulting generator matrix $G_{\D}$ is formed by permuting the columns of 
$\mathcal{A}$ and the columns of $\mathcal{B}$ separately in accordance with the ordering.

The proposition requires $q \geq 4$ to ensure $m \geq 3$. When $q = t = 4$, evaluating according to Table~\ref{table:albe},
\begin{equation*}
\mathcal{A}=
\left(\begin{array}{c c c c}
1 & 1 & 1      & 1    \\
0 & 0 & 0      & 0    \\
1 & 0 & a^{10} & a^{5} 
\end{array}\right) \text{, }
\mathcal{B}=
\left(\begin{array}{c c c c c c}
1   & 1      & 1      & 1   & 1   & 1 \\
1   & 1      & a^{10} & a^5 & a^5 & a^{10} \\
a^5 & a^{10} & 1      & 1   & a^5 & a^{10} 
\end{array}\right)\text{.}
\end{equation*}
When $t=0$, we have $G_{\D}=\mathcal{B}$. We verify using computer algebra that any three columns of $\mathcal{B}$ 
form a linearly independent set while the last four columns are linearly dependent. The code $\D^{\perp}$ is therefore 
a $[6,3,4]_{4}$-MDS code.

It is straightforward to verify that the sets
\begin{align*}
&\{\bv(1), \bw(a^{3}), \bw(a^{6})\},\{\bv(0),\bw(a^{7}),\bw(a^{11})\},\\
&\{\bv(a^5),\bw(a^{2}),\bw(a^{11})\}\text{, and }\{\bv(a^{10}), \bw(a), \bw(a^{7})\}
\end{align*}
are all linearly dependent. Therefore, $d^{\perp}=3$ for $1 \leq t \leq 4$.

Our task is slightly easier when $q =5$. The matrix $\mathcal{B}$ defined by $\{\bw(\beta_{j})\}_{j=1}^{10}$ is given by 
\begin{equation}
\mathcal{B}=\left( 
\begin{array}{cccccccccc}
1& 1& 1& 1& 1& 1& 1& 1& 1& 1\\
1& 2& 0& 1& 2& 4& 0& 4& 3& 3\\
2& 4& 3& 1& 3& 1& 2& 2& 4& 3 
\end{array}
\right). 
\end{equation}
Since it is clear that columns $1,7$, and $8$, corresponding to the set $\{\bw(a), \bw(a^{9}),\bw(a^{13})\}$, of $\mathcal{B}$ 
form a linearly dependent set, $d^{\perp}=3$ for $0 \leq t \leq q$.

Finally, for $q \geq 7$ a more general argument works in all possible cases. The trace mapping $\Tr$ is a linear transformation from $\F_{q^{2}}$ 
onto $\F_{q}$. For all $\alpha \in \F_{q}$, $\Tr(\alpha)=2 \alpha$ and for all $\beta \in \F_{q^{2}}$, we have $\Tr(\beta^{q})=\Tr(\beta)$. Hence, 
when $q$ is even, $\Tr$ maps $\F_{q}$ onto $\{0\}$, while for odd $q$, when restricted to elements of $\F_{q}$, $\Tr$ is a one-to-one mapping. 
Further, since $\Tr$ is onto, there are $q$, respectively, $q-1$, elements in $\F_{q^{2}} \setminus \F_{q}$ whose image is $1$ when $q$ is even, 
respectively, when $q$ is odd. Let $S$ be the set of all elements in $\F_{q^{2}} \setminus \F_{q}$ that $\Tr$ sends to $1$. Thus,
\begin{equation}
S \cap \{\beta_1,\beta_2,\dots,\beta_r\}= 
\begin{cases}
\{\beta_{i_1},\beta_{i_2},\dots,\beta_{i_{q/2}}\}&\text{if $q$ is even,} \\
\{\beta_{i_1},\beta_{i_2},\dots,\beta_{i_{(q-1)/2}}\}&\text{if $q$ is odd.}
\end{cases}
\end{equation}
The set of columns $\{\bw(\beta_{i_1}),\bw(\beta_{i_2}),\dots,\bw(\beta_{i_{l}})\}$ is linearly dependent if $l=q/2$ and $q$ 
is even or if $l=(q-1)/2$ and $q$ is odd. We conclude that, whenever $q\geq 7$, there always exist $3$ linearly dependent columns, making $d^{\perp}=3$. 
The proof is therefore complete.
\end{proof}

We make use of a specially designed vector space to get a subcode of dual distance $\geq 4$.
\begin{proposition}\label{3a}
Let $q \geq 4$. Consider the $\F_{q}$-vector space
\[
\W_{5}:=\langle \{1,x^{q}+x,(x^{q}+x)^2,x^{q+1},x^{2(q+1)}\} \rangle \text{.}
\]
Let $\cE$ be the evaluation code associated with $\W_{5}$. That is,
\[
\cE=\{(f(\alpha_1), \dots,f(\alpha_t),f(\beta_1), \dots,f(\beta_r))\mid f(x)\in \W_{5} \}.
\]
Then $\cE^{\perp}$ is a $[t+r, t+r-5,d^{\perp}\geq 4]_q$-code and, for all $m\geq 4$ and $0\leq \ell \leq m-1$, 
\[
\cE \subset C_{q}(t,3,2) \subset C_q(t,m,\ell)\text{.}
\]
For $(q,t)=(4,0)$, $\cE^{\perp}$ is a $[6,1,6]_{4}$-MDS code, while for $q=4$ with $1 \leq t \leq 4$ 
and for $q \geq 5$, we have $d^{\perp}=4$.
\end{proposition}

\begin{proof}
Note that, in terms of $e_{i,j}(x)$ as defined in (\ref{basis}),
\begin{equation*}
\W_{5} = \langle \{e_{0,0}(x),e_{0,1}(x),e_{0,2}(x),e_{1,1}(x),e_{2,2}(x) \} \rangle \text{.}
\end{equation*}
By consulting Figure 1, it is clear that $V_{3,2} = \W_{5} \oplus \langle \{e_{1,2}(x)\} \rangle$.
Therefore, $\cE \subset C_{q}(t,3,2) \subset C_q(t,m,\ell)$ for all $m\geq 4$ and $0\leq \ell \leq m-1$. 
When $q=4$, only the first inclusion $\cE \subset C_{q}(t,3,2)$ is valid. 
Arguing in a similar manner as in the proof of Proposition~\ref{2a} above, $G_{\cE}$ given in (\ref{CheckMat2}) 
generates $\cE$.
\begin{small}
\begin{equation}\label{CheckMat2}
G_{\cE}=\left(\begin{array}{cccccccc}
1&1&\dots&1&1&1&\dots&1\\
\noalign{\smallskip}
\alpha_1^q+\alpha_1&\alpha_2^q+\alpha_2&\dots&\alpha_t^q+\alpha_t&\beta_1^q+\beta_1&\beta_2^q+\beta_2&\dots&\beta_r^q+\beta_r\\
\noalign{\smallskip}
(\alpha_1^q+\alpha_1)^2&(\alpha_2^q+\alpha_2)^2&\dots&(\alpha_t^q+\alpha_t)^2&(\beta_1^q+\beta_1)^2&(\beta_2^q+\beta_2)^2&\dots&(\beta_r^q+\beta_r)^2\\
\noalign{\smallskip}
\alpha_1^{q+1} 	&\alpha_2^{q+1} &\dots&\alpha_t^{q+1}&\beta_1^{q+1}&\beta_2^{q+1}&\dots&\beta_r^{q+1}\\
\noalign{\smallskip}
\alpha_1^{2(q+1)} 	&\alpha_2^{2(q+1)} &\dots&\alpha_t^{2(q+1)}&\beta_1^{2(q+1)}&\beta_2^{2(q+1)}&\dots&\beta_r^{2(q+1)}
\end{array}\right)\text{.}
\end{equation}
\end{small}

The length and the dimension of $\cE^{\perp}$ can be easily verified, so we proceed to showing that $d(\cE^{\perp}) \geq 4$. Let
\begin{equation}\label{M3}
\begin{pmatrix}
1 \\ a^q+a \\ (a^q+a)^2 \\ a^{q+1} \\ a^{2(q+1)}
\end{pmatrix},
\begin{pmatrix}
1 \\ b^q+b \\ (b^q+b)^2 \\ b^{q+1} \\ b^{2(q+1)}
\end{pmatrix}, \text{ and }
\begin{pmatrix}
1 \\ c^q+c \\ (c^q+c)^2 \\ c^{q+1} \\ c^{2(q+1)}
\end{pmatrix}
\end{equation}
be any three distinct columns of $G_{\cE}$. Appendix A establishes that these columns are linearly independent.

As in the proof of Proposition~\ref{2a}, we partition $G_{\cE}$ into two matrices $\mathcal{A'}$ and $\mathcal{B'}$ according 
to the ordering of the elements $\alpha_{i}$s and $\beta_{j}$s in Table~\ref{table:albe} followed by their evaluations. Hence,
$G_{\cE}=\left(\mathcal{A'}|\mathcal{B'}\right)$. For $q=t=4$, the component matrices can be constructed explicitly as
\begin{equation}\label{d4}
\mathcal{A'}:=
\left(\begin{array}{c c c c}
1 & 1 & 1      & 1  \\
0 & 0 & 0      & 0  \\
0 & 0 & 0      & 0  \\
1 & 0 & a^{10} & a^{5} \\
1 & 0 & a^{5}  & a^{10}
\end{array}\right) \text{ and }
\mathcal{B'}:=
\left(\begin{array}{c c c c c c}
1   & 1      & 1      & 1      & 1      & 1 \\
1   & 1      & a^{10} & a^5    & a^5    & a^{10} \\
1   & 1      & a^{5}  & a^{10} & a^{10} & a^{5} \\
a^5 & a^{10} & 1      & 1   & a^5 & a^{10} \\
a^{10} & a^5 & 1      & 1   & a^{10} & a^5
\end{array}\right)\text{.}
\end{equation}

When $(q,t)=(4,0)$, a computer algebra verification certifies that every five columns of $G_{\cE}=\mathcal{B'}$ are linearly independent, 
proving that $d^{\perp}=6$.

To show that $d^{\perp}=4$ for $q=4$ when $1 \leq t \leq 4$, it is enough to exhibit that for any column of 
$\mathcal{A'}$ one can choose three columns of $\mathcal{B'}$ so that the four columns form a dependent set. 
The reader is invited to verify that the sets
\begin{align*}
&\{\bv(1), \bw(a), \bw(a^{3}), \bw(a^{7})\}, \{\bv(a^5),\bw(a^{2}),\bw(a^{3}), \bw(a^{6})\},\\ 
&\{\bv(0),\bw(a^{2}),\bw(a^{3}),\bw(a^{7})\}\text{, and }\{\bv(a^{10}), \bw(a), \bw(a^{3}), \bw(a^{6})\} 
\end{align*}
are indeed linearly dependent.

For $q=5$, the matrix $\mathcal{B'}$ is given by
\begin{equation*}
\mathcal{B'}:=
\left(\begin{array}{c c c c c c c c c c}
1 & 1 & 1 & 1 & 1 & 1 & 1 & 1 & 1 & 1 \\
1 & 2 & 0 & 1 & 2 & 4 & 0 & 4 & 3 & 3 \\
1 & 4 & 0 & 1 & 4 & 1 & 0 & 1 & 4 & 4 \\
2 & 4 & 3 & 1 & 3 & 1 & 2 & 2 & 4 & 3 \\
4 & 1 & 4 & 1 & 4 & 1 & 4 & 4 & 1 & 4 
\end{array}\right)\text{.}
\end{equation*}
Columns $1,2,8$ and $9$ are linearly dependent since 
\[
\bw(a) + 2 \bw(a^{2}) + 4 \bw(a^{13}) +3 \bw(a^{14}) = \0 \text{,}  
\]
making $d^{\perp}=4$ for all $t$.

Finally, let $q \geq 7$. The strategy is to exhibit that it is always possible to choose four columns of $\mathcal{B'}$ 
such that they form a $5\times4$ matrix of rank at most $3$. Note that the norm mapping $\N$ maps $\F_{q^{2}} \setminus\{0\}$ onto 
$\F_{q}\setminus\{0\}$ and that $\N(\alpha^{q})=\N(\alpha)$ for all $\alpha \in \F_{q^{2}}$.

By the pigeonhole principle, among the $r$ elements $\beta_{j}$s, there are always at least four elements, 
say $\beta_{j_{1}}, \beta_{j_{2}}, \beta_{j_{3}}$, and $\beta_{j_{4}}$, with 
$\N(\beta_{j_{1}})=\N(\beta_{j_{2}})=\N(\beta_{j_{3}})=\N(\beta_{j_{4}}) = \gamma \in \F_{q}\setminus \{0\}$. 
The four columns $\bw(\beta_{j_{1}}),\bw(\beta_{j_{2}}),\bw(\beta_{j_{3}}),\bw(\beta_{j_{4}})$ have the same corresponding entries in 
their last two rows and, hence, form a matrix of rank at most $3$, implying that $d^{\perp}=4$.
\end{proof}

Next, we design a subcode of an XL code with dual distance $\geq 5$.
\begin{proposition}\label{4a}
Let $q\geq 5$. Consider the $\F_{q}$-vector space
\begin{equation*}
\W_{8} = \langle \{1, x^{q+1}, x^{2(q+1)}, x^{3(q+1)}, x^{q}+x, x^{2q}+x^{2}, x^{3q}+x^3, x^{2q+1}+x^{q+2} \} \rangle \text{.} 
\end{equation*}
Let $\cF$ be the evaluation code associated with $\W_{8}$. That is,
\[
\cF=\{(f(\alpha_1), \dots,f(\alpha_t),f(\beta_1), \dots,f(\beta_r))\mid f(x)\in \W_{8} \}.
\]
Define $\V:=\W_{8} \oplus \langle \{ e_{1,3}(x)\} \rangle$ and $\V':=\W_{8} \oplus \langle \{ e_{2,3}(x)\} \rangle$. 
Let $\cF_{1}$ and $\cF_{2}$ be, respectively, the linear $q$-ary evaluation codes associated with the spaces $\V$ and $\V'$.

Then $\cF^{\perp}$ is a $[t+r, t+r-8,d^{\perp}\geq 5]_q$-code and, for all $m\geq 5$ and $0\leq \ell \leq m-1$,
\begin{align}\label{nest}
& \cF \subset \cF_{1} \subset C_{q}(t,4,3) \subset C_q(t,m,\ell) \text{, and}\notag\\
& \cF \subset \cF_{2} \subset C_{q}(t,4,3) \subset C_q(t,m,\ell) \text{.}
\end{align}
\end{proposition}

\begin{proof} 
A closer look at Figure 1 reveals that
\[
V_{4,3} = \W_{8} \oplus \langle \{e_{1,3}(x),e_{2,3}(x) \} \rangle \text{,} 
\]
justifying the nestedness presented in (\ref{nest}).

Moreover, since
\begin{equation*}
(x^{q}+x)^{2} \in \langle \{ x^{2q}+x^{2}, x^{q+1}\} \rangle \text { and } (x^{q}+x)^{3} \in \langle \{ x^{3q}+x^{3}, x^{2q+1}+x^{q+2}\} \rangle \text{,} 
\end{equation*}
we can rewrite
\begin{equation}\label{ease}
\W_{8} = \langle \{ 1, x^{q+1}, x^{2(q+1)}, x^{3(q+1)}, x^{q}+x,(x^{q}+x)^2, (x^{q}+x)^3, x^{2q+1}+x^{q+2} \} \rangle \text{.} 
\end{equation}
Evaluating based on the elements of $\W_{8}$ as expressed in (\ref{ease}) gives us a generator matrix 
$G_{\cF}$ in (\ref{CheckMat4}) of the code $\cF$. 
\begin{small}
\begin{equation}\label{CheckMat4}
G_{\cF}=\left(\begin{array}{ccccccc}
1&\dots&1&1&1&\dots&1\\
\noalign{\smallskip}
\alpha_1^{q+1}          &\dots&\alpha_t^{q+1}&\beta_1^{q+1}&\beta_2^{q+1}&\dots&\beta_r^{q+1}\\
\noalign{\smallskip}
\alpha_1^{2(q+1)}       &\dots&\alpha_t^{2(q+1)}&\beta_1^{2(q+1)}&\beta_2^{2(q+1)}&\dots&\beta_r^{2(q+1)}\\
\noalign{\smallskip}
\alpha_1^{3(q+1)}       &\dots&\alpha_t^{3(q+1)}&\beta_1^{3(q+1)}&\beta_2^{3(q+1)}&\dots&\beta_r^{3(q+1)}\\
\noalign{\smallskip}
\alpha_1^q+\alpha_1     &\dots&\alpha_t^q+\alpha_t&\beta_1^q+\beta_1&\beta_2^q+\beta_2&\dots&\beta_r^q+\beta_r\\
\noalign{\smallskip}
(\alpha_1^q+\alpha_1)^2 &\dots&(\alpha_t^q+\alpha_t)^2&(\beta_1^q+\beta_1)^2&(\beta_2^q+\beta_2)^2&\dots&(\beta_r^q+\beta_r)^2\\
\noalign{\smallskip}
(\alpha_1^q+\alpha_1)^3 &\dots&(\alpha_t^q+\alpha_t)^3&(\beta_1^q+\beta_1)^3&(\beta_2^q+\beta_2)^3&\dots&(\beta_r^q+\beta_r)^3\\
\noalign{\smallskip}
\alpha_1^{2q+1}+\alpha_1^{q+2}&\dots&\alpha_t^{2q+1}+\alpha_t^{q+2}&\beta_1^{2q+1}+\beta_1^{q+2}&\beta_2^{2q+1}+\beta_2^{q+2}&\dots&\beta_r^{2q+1}+\beta_r^{q+2}
\end{array}\right)\text{.}
\end{equation}
\end{small}
Now, choose any four distinct columns $\cC_{1}$ to $\cC_{4}$ of $G_{\cF}$, say,
\begin{small}
\begin{equation}\label{M4}
\cC_{1}:=\begin{pmatrix}
1 \\ a^{q+1} \\ a^{2(q+1)} \\ a^{3(q+1)} \\ a^q+a \\ (a^q+a)^2 \\ (a^q+a)^3 \\ a^{2q+1}+a^{q+2}
\end{pmatrix}, 
\cC_{2}:=\begin{pmatrix}
1 \\ b^{q+1} \\ b^{2(q+1)} \\ b^{3(q+1)} \\ b^q+b \\ (b^q+b)^2 \\ (b^q+b)^3 \\ b^{2q+1}+b^{q+2}
\end{pmatrix},
\cC_{3}:=\begin{pmatrix}
1 \\ c^{q+1} \\ c^{2(q+1)} \\ c^{3(q+1)} \\ c^q+c \\ (c^q+c)^2 \\ (c^q+c)^3 \\ c^{2q+1}+c^{q+2}
\end{pmatrix},
\cC_{4}:=\begin{pmatrix}
1 \\ d^{q+1} \\ d^{2(q+1)} \\ d^{3(q+1)} \\ d^q+d \\ (d^q+d)^2 \\ (d^q+d)^3 \\ d^{2q+1}+d^{q+2}
\end{pmatrix}\text{.}
\end{equation}
\end{small}

A detailed proof of the linear independence of these four columns is in Appendix B.
\end{proof}
\begin{remark}\label{rem:5}
While we are yet to supply a proof, our computation indicates that $d(\cF^{\perp})=6$ 
for $q=5$ when $t \leq 3$, for $q=7$ when $t \leq 1$, and for $q=8$ when $t=0$, while for all other cases, $d(\cF^{\perp})=5$.
\end{remark}

\section{The Resulting Asymmetric Quantum Codes}\label{sec:AQC}
Using the standard CSS construction in Theorem \ref{css} and the results obtained in Section \ref{sec:XL}, 
the following theorems can then be established.

\begin{theorem}\label{th:2}
Let $0\leq t\leq q$, $2\leq m\leq q-1$ and $0\leq \ell \leq m-1$. Then there exists an 
\begin{equation*}
[[n,k,\{d_{z} \geq \delta, d_{x}=2\}]]_q \text{-code}
\end{equation*}
with 
\begin{equation*}
n=t+(q^2-q)/2, k= m(m-1)/2+\ell \text{, and } \delta \text{ as expressed in (\ref{ddd}).} 
\end{equation*}
\end{theorem}
\begin{proof} From (\ref{codim}) and Proposition \ref{1a}, $C_q(t,1,0) \subseteq C_q(t,m,\ell)$, with 
$d((C_q(t,1,0))^\perp)=2$. Consult Theorem~\ref{lx} to compute for the required parameters.
\end{proof}

\begin{theorem} \label{th:3}
Let $q=4, 1 \leq t \leq 4, m=3, 0\leq \ell \leq 2$ or $q \geq 5, 0 \leq t \leq q, 3 \leq m \leq q-1, 0\leq \ell \leq m-1$. 
Then there exists an 
\begin{equation*}
[[n,k,\{d_{z} \geq \delta,d_{x}=3\}]]_q \text{-code}  
\end{equation*}
with 
\begin{equation*}
n=t+(q^2-q)/2, k= m(m-1)/2+\ell-2 \text{, and } \delta \text{ as computed according to (\ref{ddd}).} 
\end{equation*}
\end{theorem}

\begin{proof}
Apply Theorem~\ref{css} and infer the parameters of the resulting AQC from Proposition~\ref{2a} and Theorem~\ref{lx}.
\end{proof}

\begin{remark}
For $(q,t)=(4,0)$ the resulting AQCs have either inferior parameters to or the same parameters as those of AQMDS codes already 
discussed in~\cite{EJKL11}. 
\end{remark}

\begin{theorem} \label{th:4}
Let $q \geq 5, 0 \leq t \leq q, 4 \leq m \leq q-1, 0\leq \ell \leq m-1$. 
Then there exists an
\begin{equation*}
[[n,k,\{d_{z} \geq \delta, d_{x}=4\}]]_q \text{-code} 
\end{equation*}
with
\begin{equation*}
n=t+(q^2-q)/2 \text{ and } k= m(m-1)/2+\ell-4\text{.}
\end{equation*}
For $q=4, 1\leq t \leq 4$, there exists an AQC with parameters 
\begin{equation*}
[[t+6,1,\{d_{z} \geq \delta,d_{x}=4\}]]_{4}\text{.}
\end{equation*}
In both cases $\delta$ is as defined in (\ref{ddd}).
\end{theorem}

\begin{proof}
The assertions follows from combining Theorem~\ref{lx} and Proposition~\ref{3a}. Note that when $1 \leq t \leq 4$, the code 
$C_{4}(t,3,2)$ has dimension $6$.
\end{proof}

\begin{theorem} \label{th:5}
Let $q \geq 5, 0\leq t\leq q, 5\leq m\leq q-1$ and $0\leq \ell \leq m-1$. Then there exists an 
\begin{equation*}
[[n,k,\{d_{z} \geq \delta ,d_{x} \geq 5\}]]_q \text{-code} 
\end{equation*}
with 
\begin{equation*}
n=t+(q^2-q)/2, k= m(m-1)/2+\ell-7\text{, and } \delta \text{ as in (\ref{ddd}).}
\end{equation*}
\end{theorem}
\begin{proof}
From Proposition~\ref{4a}, we have $\cF \subseteq C_q(t,m,\ell)$, implying 
\begin{equation*}
k=\dim(C_q(t,m,\ell))-\dim(\cF) =m(m-1)/2+\ell-7 \text{.} 
\end{equation*}
The distances $d_{z}$ and $d_{x}$ are computed based on Proposition~\ref{4a} and Theorem~\ref{lx}.
\end{proof}

\begin{remark}
Related to Remark~\ref{rem:5}, computational evidences indicate that, in general, there is no gain 
in the parameters of the derived AQC in utilizing the codes $\cF_{1}$ and $\cF_{2}$ for the cases 
mentioned in the said remark. Computational results for $q \leq 9$ as will be presented in the tables below 
suggest that it suffices to consider the chain $\cF \subset C_{q}(t,4,3) \subset C_{q}(t,m,l)$ 
for all $m \geq 5$ and $0 \leq \ell \leq m-1$.
\end{remark}

For each possible combination of the relevant parameters with $3 \leq q \leq 9$, we perform the followings:
\begin{enumerate}
 \item Compute the designed distance $\delta$ of $C:=C_{q}(t,m,\ell)$ according to Theorem~\ref{lx}.
 \item Generate and store the generator matrix of $C$ based on Table~\ref{table:albe}. 
 The real distance $d(C)\geq \delta$ can then be computed and recorded.
 \item Look up in the database of MAGMA for the value of best-known dimension linear code given $(q,n,d(C))$ by 
 using the BDLC routine. If this value equals $\dim(C)$, use $C$ as $C_{2}$ in our construction.
 \item Apply the results of Subsection~\ref{sebsec:3.2} to generate the required proper subcode $C_{1}^{\perp}$.
 \item Compute for the true distance of $C_{1}$ and derive the parameters of the resulting AQCs.
 \item Use Theorem~\ref{th:opt} as a yardstick to measure how good the code $Q$ is.
\end{enumerate}

Tables~\ref{tab:3} to~\ref{tab:9b} present good pure CSS AQCs based on the Xing-Ling codes for $3 \leq q \leq 9$. 
When the code $Q$ is optimal by attaining the upper bound in Theorem~\ref{th:opt}, it is presented in bold. If 
$d_{z}=d(C_{2}) > \delta$ we list down both values in the tables. In all other instances, $d_{z}=\delta$. We exclude 
CSS asymmetric quantum MDS codes from our tables as their treatment is already available in~\cite{EJKL11}.

For $d_{x} \in \{2,3\}$ the code $Q$ presented reaches the best possible parameters given the current state of the art 
best-known classical $\F_{q}$-linear codes. 

In a few cases, for $d_{x}=2$, there are AQCs from the so-called BKLC construction in~\cite[Subsection IV.C]{EJLP12} with better 
$d_{z}$ than those derived in this paper. They are noted at the end of each of the relevant tables. The BKLC construction 
is based on the best-known linear code $C$ in terms of its minimum distance $d(C)$ when $q,n,\dim(C)$ are fixed with the 
condition that $C$ must have a codeword of weight equals its length $n$. On the other hand, there are XL codes $C_{q}(t,m,l)$ 
with parameters $[n,k,d]_{q}$ such that $d$ is the highest possible or the highest known given $(q,n,k)$. That is, they reach the 
bound for $d$ when the BKLC routine in MAGMA is called. As we have already discussed, XL codes contain the all one word $\1$, 
making them well suited for the BKLC construction of AQCs.

In the tables below, we mark the AQCs with $\star$ whenever they are derived from XL codes but are not derivable by the BKLC 
construction from the stored codes in the MAGMA database. For example, the currently stored $[7,2,5]_{4}$-code in MAGMA does 
not contain the words of weight $7$ but the XL code $C_{4}(1,2,0)$ listed as Entry 4 in Table~\ref{tab:3} contains $\1$.

There are also a couple of known instances where the CSS-like construction based on subfield linear codes as discussed in~\cite{EJLP12} 
yields AQCs with better parameters than those based on the XL codes. The note at the end of Table~\ref{tab:3} highlights two such instances. 

For $3 \leq q \leq 9$ we found exactly two instances when $C_{q}(t,3,2) \subset C_{q}(t,m,l)$ in Proposition~\ref{3a} yields AQCs with parameters 
reaching the equality in Theorem~\ref{th:opt}. They are marked with $\dag$ in Table~\ref{tab:3}. 
The AQCs marked $\#$ in Tables~\ref{tab:3} to~\ref{tab:8} are closely connected to the codes noted in Remark~\ref{rem:5}.
\begin{table}
\caption{Good $q$-ary AQCs for $q \in \{3,4,5\}$}
\label{tab:3}
\begin{center}
\setlength{\tabcolsep}{4pt}
\begin{tabular}{c c l c l l l l}
	\hline
	\noalign{\smallskip}
	No. & $q$ & $(t,m,\ell)$ & $(d_{z},\delta)$ & $Q$ from Th.~\ref{th:2} & $Q$ from Th.~\ref{th:3} & $Q$ from Th.~\ref{th:4} & $Q$ from Th.~\ref{th:5}\\
	\noalign{\smallskip}
	\hline
	\noalign{\smallskip}
	1 & $3$ & $(2,2,0)$ & & $\pmb{[[5,1,\{3,2\}]]_3}\star$ & & &\\
	2 &     & $(3,2,0)$ & & $\pmb{[[6,1,\{4,2\}]]_3}\star$ & & &\\
	3 &     & $(3,2,1)$ & & $\pmb{[[6,2,\{3,2\}]]_3}$ & & &\\
	\noalign{\smallskip}
	\hline
	\noalign{\smallskip}
	4 & $4$ & $(1,2,0)$ & & $\pmb{[[7,1,\{5,2\}]]_4}\star$ & & &\\
	5 &    & $(1,2,1)$ & & $\pmb{[[7,2,\{4,2\}]]_4}$      & & &\\
    6 &    & $(1,3,0)$ & & $\pmb{[[7,3,\{3,2\}]]_4}$ & $\pmb{[[7,1,\{3,3\}]]_4}$ & &\\
	7 &    & $(2,2,0)$ & & $\pmb{[[8,1,\{6,2\}]]_4}$    &  & &\\
	8 &    & $(2,2,1)$ & & $\pmb{[[8,2,\{5,2\}]]_4}$ & & &\\
	\noalign{\smallskip}
	9 &    & $(2,3,0)$ & & $\pmb{[[8,3,\{4,2\}]]_4}$ & $\pmb{[[8,1,\{4,3\}]]_4}$ & &\\
	10 &    & $(2,3,1)$ & & $\pmb{[[8,4,\{3,2\}]]_4}$ & $\pmb{[[8,2,\{3,3\}]]_4}$ & &\\
	11 &    & $(3,2,1)$ & & $\pmb{[[9,2,\{6,2\}]]_4}$ & & &\\
	12 &    & $(3,3,1)$ & & $\pmb{[[9,4,\{4,2\}]]_4}$ & $\pmb{[[9,2,\{4,3\}]]_4}$ & &\\
	13 &    & $(3,3,2)$ & & $\pmb{[[9,5,\{3,2\}]]_4}$ & $\pmb{[[9,3,\{3,3\}]]_4}$ & $[[9,1,\{3,4\}]]_4$ &\\
	14 &    & $(4,3,2)$ & $(4,3)$ & $\pmb{[[10,5,\{4,2\}]]_4}$ & $\pmb{[[10,3,\{4,3\}]]_4}$ & $[[10,1,\{4,4\}]]_4$ &\\
	\noalign{\smallskip}
	\hline
	\noalign{\smallskip}
	15 & $5$ & $(0,2,0)$ & & $\pmb{[[10,1,\{8,2\}]]_5}\star$ & & &\\
	16 &     & $(0,2,1)$ & & $\pmb{[[10,2,\{7,2\}]]_5}$ & & &\\
	17 &     & $(0,3,1)$ & & $\pmb{[[10,4,\{5,2\}]]_5}$ & $\pmb{[[10,2,\{5,3\}]]_5}$ & &\\
	18 &     & $(0,3,2)$ & & $\pmb{[[10,5,\{4,2\}]]_5}$ & $\pmb{[[10,3,\{4,3\}]]_5}$ & &\\
 	19 &     & $(0,4,0)$ & & $\pmb{[[10,6,\{3,2\}]]_5}$ & $\pmb{[[10,4,\{3,3\}]]_5}$ & $[[10,2,\{3,4\}]]_5$ &\\
	20 &     &           & &                       & & $\pmb{[[10,1,\{3,6\}]]_5} \dag$ & \\
\noalign{\smallskip}
	21 &     & $(1,2,1)$ & & $\pmb{[[11,2,\{8,2\}]]_5}$ & & &\\
	22 &     & $(1,3,2)$ & & $\pmb{[[11,5,\{5,2\}]]_5}$ & $\pmb{[[11,3,\{5,3\}]]_5}$ & $[[11,1,\{5,4\}]]_5$ &\\
	23 &     & $(1,4,1)$ & & $\pmb{[[11,7,\{3,2\}]]_5}$ & $\pmb{[[11,5,\{3,3\}]]_5}$ & $[[11,3,\{3,4\}]]_5$ &\\
	24 &     &           &                                   & & & $\pmb{[[11,2,\{3,6\}]]_5} \dag$ & \\
	25 & & $(2,2,0)$ & & $\pmb{[[12,1,\{9,2\}]]_5}\star$  & & &\\
\noalign{\smallskip}
	26 & & $(2,4,2)$ & & $\pmb{[[12,8,\{3,2\}]]_5}$  & $\pmb{[[12,6,\{3,3\}]]_5}$ & $[[12,4,\{3,4\}]]_5$ &\\
	27 & & $(3,2,0)$ & & $\pmb{[[13,1,\{10,2\}]]_5}$ & & &\\
	28 & & $(3,2,1)$ & & $\pmb{[[13,2,\{9,2\}]]_5}$  & & &\\
	29 & & $(3,3,2)$ & & $\pmb{[[13,5,\{6,2\}]]_5}\star$ & $\pmb{[[13,3,\{6,3\}]]_5}$ & $[[13,1,\{6,4\}]]_5$ &\\
	30 & & $(3,4,3)$ & & $\pmb{[[13,9,\{3,2\}]]_5}$ & $\pmb{[[13,7,\{3,3\}]]_5}$ & $[[13,5,\{3,4\}]]_5$ & $[[13,2,\{3,6\}]]_5\#$\\
\noalign{\smallskip}
	31 & & $(4,2,0)$ &  & $\pmb{[[14,1,\{11,2\}]]_5}$ & & &\\
	32 & & $(4,2,1)$ &  & $\pmb{[[14,2,\{10,2\}]]_5}$ & & &\\
	33 & & $(4,4,1)$ & $(5,4)$ & $\pmb{[[14,7,\{5,2\}]]_5}$  & $\pmb{[[14,5,\{5,3\}]]_5}$ & $[[14,3,\{5,4\}]]_5$ &\\
	34 & & $(5,2,0)$ &  & $\pmb{[[15,1,\{12,2\}]]_5}\star$ & & &\\
	35 & & $(5,2,1)$ &  & $\pmb{[[15,2,\{11,2\}]]_5}$ & & &\\
	36 & & $(5,4,1)$ & $(6,5)$ & $\pmb{[[15,7,\{6,2\}]]_5}$  & $\pmb{[[15,5,\{6,3\}]]_5}$ & $[[15,3,\{6,4\}]]_5$ &\\
\noalign{\smallskip}
\hline
\end{tabular}
\end{center}
\vspace{-8pt}
\scriptsize
\begin{tabular}{l l  l l}
\textbf{Note on} & \textbf{Remark} & \textbf{Note on} & \textbf{Remark}\\
Entry 4 & $\exists~[[7,1.5,\{5,2\}]]_4$ CSS-like in~\cite{EJLP12} & Entry 20 & Use $C_{5}(0,3,2) \subset C_{5}(0,4,0)$ \\
Entry 6 & $\exists~[[6,1,\{3,3\}]]_4$ CSS-like in~\cite{EJLP12}   & Entry 24 & Use $C_{5}(1,3,2) \subset C_{5}(1,4,1)$\\
\hline
\end{tabular}
\end{table}
\begin{table}[!hbt]
\caption{Good $7$-ary AQCs}
\label{tab:7}
\begin{center}
\setlength{\tabcolsep}{4pt}
\begin{tabular}{c l c l l l l}
\hline
\noalign{\smallskip}
No. & $(t,m,\ell)$ & $(d_{z},\delta)$ & $Q$ from Th.~\ref{th:2} & $Q$ from Th.~\ref{th:3} & $Q$ from Th.~\ref{th:4} & $Q$ from Th.~\ref{th:5} \\
\noalign{\smallskip}
\hline
\noalign{\smallskip}
1 & $(0,2,0)$ &            & $\pmb{[[21,1,\{18,2\}]]_7}\star$  & & & \\
2 & $(0,2,1)$ &            & $\pmb{[[21,2,\{17,2\}]]_7}$       &  & & \\
3 & $(0,3,1)$ &            & $[[21,4,\{14,2\}]]_7$ & $[[21,2,\{14,3\}]]_7$ & & \\
4 & $(0,3,2)$ &            & $[[21,5,\{13,2\}]]_7$ & $[[21,3,\{13,3\}]]_7$ & $[[21,1,\{13,4\}]]_7$ & \\
5 & $(0,4,0)$ & $(12,11)$  & $[[21,6,\{12,2\}]]_7$ & $[[21,4,\{12,3\}]]_7$ & $[[21,2,\{12,4\}]]_7$ & \\
\noalign{\smallskip}
6 & $(0,4,2)$ & & $[[21,8,\{10,2\}]]_7$ & $[[21,6,\{10,3\}]]_7$ & $[[21,4,\{10,4\}]]_7 $ & \\
7 & $(0,4,3)$ & & $[[21,9,\{9,2\}]]_7$ & $[[21,7,\{9,3\}]]_7$   & $[[21,5,\{9,4\}]]_7 $ & $[[21,2,\{9,6\}]]_7 \# $ \\
8 & $(0,5,1)$ & & $[[21,11,\{7,2\}]]_7$ & $[[21,9,\{7,3\}]]_7$  & $[[21,7,\{7,4\}]]_7$ & $[[21,4,\{7,6\}]]_7 \# $ \\
9 & $(0,5,3)$ & & $[[21,13,\{6,2\}]]_7$ & $[[21,11,\{6,3\}]]_7$ & $[[21,9,\{6,4\}]]_7$ & $[[21,6,\{6,6\}]]_7 \# $ \\
10 & $(0,5,4)$ & & $[[21,14,\{5,2\}]]_7$ & $[[21,12,\{5,3\}]]_7$ & $[[21,10,\{5,4\}]]_7$& $[[21,7,\{5,6\}]]_7 \# $ \\
\noalign{\smallskip}
11 & $(0,6,2)$ &           & $\pmb{[[21,17,\{3,2\}]]_7}$ & $\pmb{[[21,15,\{3,3\}]]_7}$ & $[[21,13,\{3,4\}]]_7$ & $[[21,10,\{3,6\}]]_7 \# $ \\
12 & $(1,2,1)$ &           & $\pmb{[[22,2,\{18,2\}]]_7}$ &  & & \\
13 & $(1,3,2)$ &           & $[[22,5,\{14,2\}]]_7$              & $[[22,3,\{14,3\}]]_7$ & $[[22,1,\{14,4\}]]_7$ & \\
14 & $(1,4,0)$ & $(12,11)$ & $[[22,6,\{12,2\}]]_7$              & $[[22,4,\{12,3\}]]_7$ & $[[22,2,\{12,4\}]]_7$ & \\
15 & $(1,4,1)$ &           & $[[22,7,\{11,2\}]]_7$              & $[[22,5,\{11,3\}]]_7$ & $[[22,3,\{11,4\}]]_7$ & \\
\noalign{\smallskip}
16 & $(1,4,3)$ & & $[[22,9,\{10,2\}]]_7$ & $[[22,7,\{10,3\}]]_7$ & $[[22,5,\{10,4\}]]_7$ & $[[22,2,\{10,6\}]]_7 \# $ \\
17 & $(1,5,2)$ & & $[[22,12,\{7,2\}]]_7$ & $[[22,10,\{7,3\}]]_7$ & $[[22,8,\{7,4\}]]_7$  & $[[22,5,\{7,6\}]]_7 \# $ \\
18 & $(1,5,4)$ & & $[[22,14,\{6,2\}]]_7$ & $[[22,12,\{6,3\}]]_7$ & $[[22,10,\{6,4\}]]_7$ & $[[22,7,\{6,6\}]]_7 \# $ \\
19 & $(1,6,3)$ & & $\pmb{[[22,18,\{3,2\}]]_7}$ & $\pmb{[[22,16,\{3,3\}]]_7}$ & $[[22,14,\{3,4\}]]_7$ & $[[22,11,\{3,6\}]]_7 \#$ \\
20 & $(2,2,0)$ & & $\pmb{[[23,1,\{19,2\}]]_7}\star$ &  & &  \\
\noalign{\smallskip}
21	& $(2,2,1)$ & & $[[23,2,\{18,2\}]]_7$ & & & \\
22	& $(2,3,2)$ & & $[[23,5,\{14,2\}]]_7$ & $[[23,3,\{14,3\}]]_7$ & $[[23,1,\{14,4\}]]_7$ & \\ 
23	& $(2,4,2)$ & & $[[23,8,\{11,2\}]]_7$ & $[[23,6,\{11,3\}]]_7$ & $[[23,4,\{11,4\}]]_7$ & \\
24	& $(2,5,3)$ & & $[[23,13,\{7,2\}]]_7$ & $[[23,11,\{7,3\}]]_7$ & $[[23,9,\{7,4\}]]_7$ & $[[23,6,\{7,5\}]]_7$ \\
25	& $(2,6,4)$ & & $\pmb{[[23,19,\{3,2\}]]_7}$ & $\pmb{[[23,17,\{3,3\}]]_7}$ & $[[23,15,\{3,4\}]]_7$ & $[[23,12,\{3,5\}]]_7$\\
\noalign{\smallskip}
26	& $(3,2,0)$ & & $\pmb{[[24,1,\{20,2\}]]_7}\star$ &  & & \\
27	& $(3,2,1)$ & & $\pmb{[[24,2,\{19,2\}]]_7}$  &  & & \\
28	& $(3,3,2)$ & & $[[24,5,\{15,2\}]]_7$ & $[[24,3,\{15,3\}]]_7$ & $[[24,1,\{15,4\}]]_7$ & \\
29	& $(3,4,3)$ & & $[[24,9,\{11,2\}]]_7$ & $[[24,7,\{11,3\}]]_7$ & $[[24,5,\{11,4\}]]_7$ & $[[24,2,\{11,5\}]]_7$\\
30	& $(3,5,4)$ & & $[[24,14,\{7,2\}]]_7$ & $[[24,12,\{7,3\}]]_7$ & $[[24,10,\{7,4\}]]_7$ & $[[24,7,\{7,5\}]]_7$\\
\noalign{\smallskip}
31	& $(3,6,5)$ & & $\pmb{[[24,20,\{3,2\}]]_7}$ & $\pmb{[[24,18,\{3,3\}]]_7}$ & $[[24,16,\{3,4\}]]_7$ & $[[24,13,\{3,5\}]]_7$\\
32	& $(4,2,0)$ & & $\pmb{[[25,1,\{21,2\}]]_7}\star$ &  & &\\
33	& $(4,2,1)$ & & $\pmb{[[25,2,\{20,2\}]]_7}$ &  &  &\\
34	& $(4,5,3)$ & & $[[25,13,\{8,2\}]]_7$ & $[[25,11,\{8,3\}]]_7$ & $[[25,9,\{8,4\}]]_7$ & $[[25,6,\{8,5\}]]_7$\\
35	& $(5,2,0)$ & & $\pmb{[[26,1,\{22,2\}]]_7}\star$  & & &\\
\noalign{\smallskip}
36	& $(5,2,1)$ & & $\pmb{[[26,2,\{21,2\}]]_7}$  &  & &\\ 
37  & $(5,5,4)$ & & $[[26,14,\{8,2\}]]_7$ & $[[26,12,\{8,3\}]]_7$ & $[[26,10,\{8,4\}]]_7$ & $[[26,7,\{8,5\}]]_7$\\
38	& $(6,2,0)$ & & $\pmb{[[27,1,\{23,2\}]]_7}\star$  &  & &\\
39	& $(6,2,1)$ & & $\pmb{[[27,2,\{22,2\}]]_7}$  &  & &\\
40	& $(7,2,0)$ & & $\pmb{[[28,1,\{24,2\}]]_7}\star$  & & &\\
\noalign{\smallskip}
41	& $(7,2,1)$ & & $\pmb{[[28,2,\{23,2\}]]_7}$  &  & & \\
42	& $(7,4,2)$ & & $[[28,8,\{14,2\}]]_7$ & $[[28,6,\{14,3\}]]_7$ & $ [[28,4,\{14,4\}]]_7$ & \\
43	& $(7,6,2)$ & $(7,6)$ & $[[28,17,\{7,2\}]]_7$ & $[[28,15,\{7,3\}]]_7$ & $[[28,13,\{7,4\}]]_7$ & $[[28,10,\{7,5\}]]_7$\\
\noalign{\smallskip}
\hline
\end{tabular}
\end{center}
\vspace{-8pt}
\scriptsize
\begin{tabular}{l l  l l}
\textbf{Note on} & \textbf{Remark} & \textbf{Note on} & \textbf{Remark}\\
Entry 14 & $\exists [[22,6,\{13,2\}]]_7$ BKLC & Entry 28 & $\exists [[24,5,\{16,2\}]]_7$ BKLC \\
Entry 22 & $\exists [[23,5,\{15,2\}]]_7$ BKLC & Entry 42 & $\exists [[28,8,\{15,2\}]]_7$ BKLC \\
\hline
\end{tabular}
\end{table}
\begin{table}
\caption{Good $8$-ary AQCs}
\label{tab:8}
\vspace{-12pt}
\begin{center}
\setlength{\tabcolsep}{4pt}
\begin{tabular}{c l c l l l l}
\hline
\noalign{\smallskip}
No. & $(t,m,\ell)$ & $(d_{z},\delta)$ & $Q$ from Th.~\ref{th:2} & $Q$ from Th.~\ref{th:3} & $Q$ from Th.~\ref{th:4} & $Q$ from Th.~\ref{th:5} \\
\noalign{\smallskip}
\hline
\noalign{\smallskip}
1 & $(0,2,1)$ & & $\pmb{[[28,2,\{24,2\}]]_8}$ & &  &\\
2 & $(0,3,1)$ & & $[[28,4,\{20,2\}]]_8$ & $[[28,2,\{20,3\}]]_8$ & &\\
3 & $(0,3,2)$ & & $[[28,5,\{19,2\}]]_8$ & $[[28,3,\{19,3\}]]_8$ & $[[28,1,\{19,4\}]]_8$ &\\
4 & $(0,4,3)$ & & $[[28,9,\{15,2\}]]_8$ & $[[28,7,\{15,3\}]]_8$ & $[[28,5,\{15,4\}]]_8$ & $[[28,2,\{15,6\}]]_8 \# $\\
5 & $(0,5,3)$ & & $[[28,13,\{11,2\}]]_8$ & $[[28,11,\{11,3\}]]_8$ & $[[28,9,\{11,4\}]]_8$ & $[[28,6,\{11,6\}]]_8 \# $\\
\noalign{\smallskip}
6 & $(0,5,4)$ & & $[[28,14,\{10,2\}]]_8$ & $[[28,12,\{10,3\}]]_8$ & $[[28,10,\{10,4\}]]_8$ & $[[28,7,\{10,6\}]]_8 \# $\\
7 & $(0,6,3)$ & & $[[28,18,\{7,2\}]]_8$  & $[[28,16,\{7,3\}]]_8$ & $[[28,14,\{7,4\}]]_8$ & $[[28,11,\{7,6\}]]_8 \# $\\
8 & $(0,6,5)$ & & $[[28,20,\{6,2\}]]_8$  & $[[28,18,\{6,3\}]]_8$ & $[[28,16,\{6,4\}]]_8$ & $[[28,13,\{6,6\}]]_8 \# $\\
9 & $(0,7,3)$ & & $\pmb{[[28,24,\{3,2\}]]_8}$ & $\pmb{[[28,22,\{3,3\}]]_8}$ & $[[28,20,\{3,4\}]]_8$ & $[[28,17,\{3,6\}]]_8 \# $\\
10 & $(1,2,0)$ & & $\pmb{[[29,1,\{25,2\}]]_8}\star$ & & &\\
\noalign{\smallskip}
11 & $(1,2,1)$ & & $\pmb{[[29,2,\{24,2\}]]_8}$ &  & &\\
12 & $(1,3,2)$ & & $[[29,5,\{20,2\}]]_8$ & $[[29,3,\{20,3\}]]_8$ & $[[29,1,\{20,4\}]]_8$ &\\
13 & $(1,4,3)$ & & $[[29,9,\{15,2\}]]_8$ & $[[29,7,\{15,3\}]]_8$ & $[[29,5,\{15,4\}]]_8$ & $[[29,2,\{15,5\}]]_8$\\
14 & $(1,5,4)$ & & $[[29,14,\{11,2\}]]_8$ & $[[29,12,\{11,3\}]]_8$ & $[[29,10,\{11,4\}]]_8$ & $[[29,7,\{11,5\}]]_8$\\
15 & $(1,6,2)$ & & $[[29,17,\{8,2\}]]_8$ & $[[29,15,\{8,3\}]]_8$ & $[[29,13,\{8,4\}]]_8$ & $[[29,10,\{8,5\}]]_8$\\
\noalign{\smallskip}
16 & $(1,6,4)$ & & $[[29,19,\{7,2\}]]_8$ & $[[29,17,\{7,3\}]]_8$ & $[[29,15,\{7,4\}]]_8$ & $[[29,12,\{7,5\}]]_8$\\
17 & $(1,7,4)$ & & $\pmb{[[29,25,\{3,2\}]]_8}$ & $\pmb{[[29,23,\{3,3\}]]_8}$ & $[[29,21,\{3,4\}]]_8$ & $[[29,18,\{3,5\}]]_8$\\
18 & $(2,2,0)$ & & $\pmb{[[30,1,\{26,2\}]]_8}\star$ &  & &\\
19 & $(2,2,1)$ & & $\pmb{[[30,2,\{25,2\}]]_8}$ &  & &\\
20 & $(2,3,1)$ & & $[[30,4,\{21,2\}]]_8$ & $[[30,2,\{21,3\}]]_8$ & & \\
\noalign{\smallskip}
21 & $(2,3,2)$ & & $[[30,5,\{20,2\}]]_8$ & $[[30,3,\{20,3\}]]_8$ & $[[30,1,\{20,4\}]]_8$ &\\
22 & $(2,4,3)$ & & $[[30,9,\{16,2\}]]_8$ & $[[30,7,\{16,3\}]]_8$ & $[[30,5,\{16,4\}]]_8$ & $[[30,2,\{16,5\}]]_8$\\
23 & $(2,5,3)$ & & $[[30,13,\{12,2\}]]_8$ & $[[30,11,\{12,3\}]]_8$ & $[[30,9,\{12,4\}]]_8$ & $[[30,6,\{12,5\}]]_8$\\
24 & $(2,5,4)$ & & $[[30,14,\{11,2\}]]_8$ & $[[30,12,\{11,3\}]]_8$ & $[[30,10,\{11,4\}]]_8$ & $[[30,7,\{11,5\}]]_8$\\
25 & $(2,6,3)$ & & $[[30,18,\{8,2\}]]_8$ & $[[30,16,\{8,3\}]]_8$ & $[[30,14,\{8,4\}]]_8$ & $[[30,11,\{8,5\}]]_8$\\
\noalign{\smallskip}
26 & $(2,6,5)$ & & $[[30,20,\{7,2\}]]_8$ & $[[30,18,\{7,3\}]]_8$ & $[[30,16,\{7,4\}]]_8$ & $[[30,13,\{7,5\}]]_8$\\
27 & $(2,7,5)$ & & $\pmb{[[30,26,\{3,2\}]]_8}$ & $\pmb{[[30,24,\{3,3\}]]_8}$ & $[[30,22,\{3,4\}]]_8$ & $[[30,19,\{3,5\}]]_8$\\
28 & $(3,2,0)$ & & $\pmb{[[31,1,\{27,2\}]]_8}\star$ & & &\\
29 & $(3,2,1)$ & & $\pmb{[[31,2,\{26,2\}]]_8}$ &  & &\\
30 & $(3,3,1)$ & & $[[31,4,\{22,2\}]]_8$ & $[[31,2,\{22,3\}]]_8$ & & \\
\noalign{\smallskip}
31 & $(3,3,2)$ & & $[[31,5,\{21,2\}]]_8$ & $[[31,3,\{21,3\}]]_8$ & $[[31,1,\{21,4\}]]_8$ &\\
32 & $(3,4,1)$ & & $[[31,7,\{18,2\}]]_8$ & $[[31,5,\{18,3\}]]_8$ & $[[31,3,\{18,4\}]]_8$ & \\
33 & $(3,4,2)$ & & $[[31,8,\{17,2\}]]_8$ & $[[31,6,\{17,3\}]]_8$ & $[[31,4,\{17,4\}]]_8$ &\\
34 & $(3,4,3)$ & & $[[31,9,\{16,2\}]]_8$ & $[[31,7,\{16,3\}]]_8$ & $[[31,5,\{16,4\}]]_8$ & $[[31,2,\{16,5\}]]_8$\\
35 & $(3,5,0)$ & $(15,14)$ & $[[31,10,\{15,2\}]]_8$ & $[[31,8,\{15,3\}]]_8$ & $[[31,6,\{15,4\}]]_8$ & $[[31,3,\{15,5\}]]_8$ \\
\noalign{\smallskip}
36 & $(3,5,2)$ & & $[[31,12,\{13,2\}]]_8$ & $[[31,10,\{13,3\}]]_8$ & $[[31,8,\{13,4\}]]_8$ & $[[31,5,\{13,5\}]]_8$\\
37 & $(3,5,4)$ & & $[[31,14,\{12,2\}]]_8$ & $[[31,12,\{12,3\}]]_8$ & $[[31,10,\{12,4\}]]_8$ & $[[31,7,\{12,5\}]]_8$\\
38 & $(3,6,4)$ & & $[[31,19,\{8,2\}]]_8$ & $[[31,17,\{8,3\}]]_8$ & $[[31,15,\{8,4\}]]_8$ & $[[31,12,\{8,5\}]]_8$\\
39 & $(3,7,6)$ & & $\pmb{[[31,27,\{3,2\}]]_8}$ & $\pmb{[[31,25,\{3,3\}]]_8}$ & $[[31,23,\{3,4\}]]_8$ & $[[31,20,\{3,5\}]]_8$\\
40 & $(4,2,0)$ & & $\pmb{[[32,1,\{28,2\}]]_8}\star$ & & &\\
\noalign{\smallskip}
41 & $(4,2,1)$ & & $\pmb{[[32,2,\{27,2\}]]_8}$ & & &\\
42 & $(4,3,1)$ & & $[[32,4,\{23,2\}]]_8$ & $[[32,2,\{23,3\}]]_8$ & & \\
43 & $(4,3,2)$ & & $[[32,5,\{22,2\}]]_8$ & $[[32,3,\{22,3\}]]_8$ & $[[32,1,\{22,4\}]]_8$ &\\
44 & $(4,4,1)$ & & $[[32,7,\{19,2\}]]_8$ & $[[32,5,\{19,3\}]]_8$ & $[[32,3,\{19,4\}]]_8$ & \\
45 & $(4,4,2)$ & & $[[32,8,\{18,2\}]]_8$ & $[[32,6,\{18,3\}]]_8$ & $[[32,4,\{18,4\}]]_8$ &\\
\noalign{\smallskip}
46 & $(4,4,3)$ & & $[[32,9,\{17,2\}]]_8$ & $[[32,7,\{17,3\}]]_8$ & $[[32,5,\{17,4\}]]_8$ & $[[32,2,\{17,5\}]]_8$\\
47 & $(4,5,0)$ & $(16,15)$ & $[[32,10,\{16,2\}]]_8$ & $[[32,8,\{16,3\}]]_8$ & $[[32,6,\{16,4\}]]_8$ & $[[32,3,\{16,5\}]]_8$ \\
48 & $(4,5,2)$ & & $[[32,12,\{14,2\}]]_8$ & $[[32,10,\{14,3\}]]_8$ & $[[32,8,\{14,4\}]]_8$ & $[[32,5,\{14,5\}]]_8$\\
49 & $(4,5,3)$ & & $[[32,13,\{13,2\}]]_8$ & $[[32,11,\{13,3\}]]_8$ & $[[32,9,\{13,4\}]]_8$ & $[[32,6,\{13,5\}]]_8$\\
50 & $(4,6,5)$ & & $[[32,20,\{8,2\}]]_8$  & $[[32,18,\{8,3\}]]_8$ & $[[32,16,\{8,4\}]]_8$ & $[[32,13,\{8,5\}]]_8$\\
\noalign{\smallskip}		
51 & $(5,2,1)$ & & $\pmb{[[33,2,\{28,2\}]]_8}$ & & &\\
52 & $(5,3,1)$ & & $[[33,4,\{24,2\}]]_8$ & $[[33,2,\{24,3\}]]_8$ & &\\
53 & $(5,3,2)$ & & $[[33,5,\{23,2\}]]_8$ & $[[33,3,\{23,3\}]]_8$ & $[[33,1,\{23,4\}]]_8$ &\\
\noalign{\smallskip}
\hline
\end{tabular}
\end{center}
\vspace{-8pt}
\scriptsize
\setlength{\tabcolsep}{4pt}
\begin{tabular}{l l  l l l l}
\textbf{Note on} & \textbf{Remark} & \textbf{Note on} & \textbf{Remark} & \textbf{Note on} & \textbf{Remark}\\
Entry 13 & $\exists [[29,9,\{16,2\}]]_8$ BKLC & Entry 30 & $\exists [[31,4,\{23,2\}]]_8$ BKLC & Entry 42 & $\exists [[32,4,\{24,2\}]]_8$ BKLC\\
Entry 20 & $\exists [[30,4,\{22,2\}]]_8$ BKLC & Entry 32 & $\exists [[31,7,\{19,2\}]]_8$ BKLC & Entry 44 & $\exists [[32,7,\{20,2\}]]_8$ BKLC\\ 
\hline
\end{tabular}
\end{table}
\begin{table}
\renewcommand\thetable{4}
\caption{Good $8$-ary AQCs {\it(Continued)}}
\label{tab:8b}
\centering
\setlength{\tabcolsep}{4pt}
\begin{tabular}{c l c l l l l}
\hline
\noalign{\smallskip}
No. & $(t,m,\ell)$ & $(d_{z},\delta)$ & $Q$ from Th.~\ref{th:2} & $Q$ from Th.~\ref{th:3} & $Q$ from Th.~\ref{th:4} & $Q$ from Th.~\ref{th:5} \\
\noalign{\smallskip}
\hline
\noalign{\smallskip}
54 & $(5,4,1)$ & & $[[33,7,\{20,2\}]]_8$ & $[[33,5,\{20,3\}]]_8$ & $[[33,3,\{20,4\}]]_8$ &\\
55 & $(5,4,2)$ & & $[[33,8,\{19,2\}]]_8$ & $[[33,6,\{19,3\}]]_8$ & $[[33,4,\{19,4\}]]_8$ &\\
56 & $(5,4,3)$ & & $[[33,9,\{18,2\}]]_8$ & $[[33,7,\{18,3\}]]_8$ & $[[33,5,\{18,4\}]]_8$ & $[[33,2,\{18,5\}]]_8$\\
57 & $(5,5,2)$ & & $[[33,12,\{15,2\}]]_8$ & $[[33,10,\{15,3\}]]_8$ & $[[33,8,\{15,4\}]]_8$ & $[[33,5,\{15,5\}]]_8$\\
58 & $(5,5,3)$ & & $[[33,13,\{14,2\}]]_8$ & $[[33,11,\{14,3\}]]_8$ & $[[33,9,\{14,4\}]]_8$ & $[[33,6,\{14,5\}]]_8$\\
\noalign{\smallskip}
59 & $(5,5,4)$ & & $[[33,14,\{13,2\}]]_8$ & $[[33,12,\{13,3\}]]_8$ & $[[33,10,\{13,4\}]]_8$ & $[[33,7,\{13,5\}]]_8$\\
60 & $(5,6,5)$ & & $\pmb{[[33,20,\{9,2\}]]_8}$ & $\pmb{[[33,18,\{8,3\}]]_8}$ & $[[33,16,\{8,4\}]]_8$ & $[[33,13,\{8,5\}]]_8$\\
61 & $(6,2,1)$ & & $\pmb{[[34,2,\{28,2\}]]_8}$ & & &\\
62 & $(6,3,2)$ & $(24,23)$ & $\pmb{[[34,5,\{24,2\}]]_8}$ & $\pmb{[[34,3,\{24,3\}]]_8}$ & $[[34,1,\{24,4\}]]_8$ & \\
63 & $(6,4,3)$ & & $\pmb{[[34,9,\{19,2\}]]_8}$ & $\pmb{[[34,7,\{19,3\}]]_8}$ & $[[34,5,\{19,4\}]]_8$ & $[[34,2,\{19,5\}]]_8$ \\
\noalign{\smallskip}
64 & $(6,5,1)$ & & $\pmb{[[34,11,\{16,2\}]]_8}$ & $\pmb{[[34,9,\{16,3\}]]_8}$ & $[[34,7,\{16,4\}]]_8$ & $[[34,4,\{16,5\}]]_8$ \\
65 & $(6,5,3)$ & & $\pmb{[[34,13,\{15,2\}]]_8}$ & $\pmb{[[34,11,\{15,3\}]]_8}$ & $[[34,9,\{15,4\}]]_8$ & $[[34,6,\{15,5\}]]_8$ \\
66 & $(6,6,1)$ & & $[[34,16,\{12,2\}]]_8$ & $[[34,14,\{12,3\}]]_8$ & $[[34,12,\{12,4\}]]_8$ & $[[34,9,\{12,5\}]]_8$ \\
67 & $(6,6,3)$ & & $[[34,18,\{11,2\}]]_8$ & $[[34,16,\{11,3\}]]_8$ & $[[34,14,\{11,4\}]]_8$ & $[[34,11,\{11,5\}]]_8$ \\
68 & $(6,7,2)$ & & $[[34,23,\{7,2\}]]_8$ & $[[34,21,\{7,3\}]]_8$ & $[[34,19,\{7,4\}]]_8$ & $[[34,16,\{7,5\}]]_8$\\
\noalign{\smallskip}
69 & $(7,3,2)$ & $(24,23)$ & $[[35,5,\{24,2\}]]_8$ & $[[35,3,\{24,3\}]]_8$ & $[[35,1,\{24,4\}]]_8$ & \\
70 & $(7,4,3)$ & & $[[35,9,\{19,2\}]]_8$ & $[[35,7,\{19,3\}]]_8$ & $[[35,5,\{19,4\}]]_8$ & $[[35,2,\{19,5\}]]_8$ \\
71 & $(7,5,4)$ & $(15,14)$ & $[[35,14,\{15,2\}]]_8$ & $[[35,12,\{15,3\}]]_8$ & $[[35,10,\{15,4\}]]_8$ & $[[35,7,\{15,5\}]]_8$ \\
72 & $(7,7,3)$ & & $[[35,24,\{7,2\}]]_8$ & $[[35,22,\{7,3\}]]_8$ & $[[35,20,\{7,4\}]]_8$ & $[[35,17,\{7,5\}]]_8$\\
73 & $(8,5,4)$ & $(15,14)$ & $[[36,14,\{15,2\}]]_8$ & $[[36,12,\{15,3\}]]_8$ & $[[36,10,\{15,4\}]]_8$ & $[[36,7,\{15,5\}]]_8$ \\
\noalign{\smallskip}
\hline
\end{tabular}
\end{table}
\begin{table}
\renewcommand\thetable{5}
\caption{Good $9$-ary AQCs}
\label{tab:9}
\centering
\setlength{\tabcolsep}{4pt}
\begin{tabular}{c l c l l l l}
\hline
\noalign{\smallskip}
No. & $(t,m,\ell)$ & $(d_{z},\delta)$ & $Q$ from Th.~\ref{th:2} & $Q$ from Th.~\ref{th:3} & $Q$ from Th.~\ref{th:4} & $Q$ from Th.~\ref{th:5} \\
\noalign{\smallskip}
\hline
\noalign{\smallskip}
1 & $(0,2,0)$ & & $\pmb{[[36,1,\{32,2\}]]_9}\star$ & & &\\   
2 & $(0,2,1)$ & & $\pmb{[[36,2,\{31,2\}]]_9}$   & & &\\
3 & $(0,3,1)$ & & $[[36,4,\{27,2\}]]_9$  & $[[36,2,\{27,3\}]]_9$ & &\\
4 & $(0,3,2)$ & & $[[36,5,\{26,2\}]]_9$  & $[[36,3,\{26,3\}]]_9$  & $[[36,1,\{26,4\}]]_9$ &\\
5 & $(0,4,0)$ & $(24,23)$ & $[[36,6,\{24,2\}]]_9$  & $[[36,4,\{24,3\}]]_9$  & $[[36,2,\{24,4\}]]_9$ &\\
\noalign{\smallskip}
6 & $(0,4,2)$ & & $[[36,8,\{22,2\}]]_9$  & $[[36,6,\{22,3\}]]_9$  & $[[36,4,\{22,4\}]]_9$ & $[[36,1,\{22,5\}]]_9$\\
7 & $(0,4,3)$ & & $[[36,9,\{21,2\}]]_9$  & $[[36,7,\{21,3\}]]_9$  & $[[36,5,\{21,4\}]]_9$ & $[[36,2,\{21,5\}]]_9$\\
8 & $(0,5,3)$ & & $[[36,13,\{17,2\}]]_9$  & $[[36,11,\{17,3\}]]_9$  & $[[36,9,\{17,4\}]]_9$ & $[[36,6,\{17,5\}]]_9$\\
9 & $(0,5,4)$ & & $[[36,14,\{16,2\}]]_9$  & $[[36,12,\{16,3\}]]_9$  & $[[36,10,\{16,4\}]]_9$ & $[[36,7,\{16,5\}]]_9$\\
10 & $(0,6,2)$ & & $[[36,17,\{13,2\}]]_9$  & $[[36,15,\{13,3\}]]_9$  & $[[36,13,\{13,4\}]]_9$ & $[[36,10,\{13,5\}]]_9$\\
\noalign{\smallskip}
11 & $(0,6,4)$ & & $[[36,19,\{12,2\}]]_9$  & $[[36,17,\{12,3\}]]_9$  & $[[36,15,\{12,4\}]]_9$ & $[[36,12,\{12,5\}]]_9$\\
12 & $(0,6,5)$ & & $[[36,20,\{11,2\}]]_9$  & $[[36,18,\{11,3\}]]_9$  & $[[36,16,\{11,4\}]]_9$ & $[[36,13,\{11,5\}]]_9$\\
13 & $(0,7,3)$ & & $[[36,24,\{8,2\}]]_9$  & $[[36,22,\{8,3\}]]_9$ & $[[36,20,\{8,4\}]]_9$ & $[[36,17,\{8,5\}]]_9$\\
14 & $(0,7,5)$ & & $[[36,26,\{7,2\}]]_9$  & $[[36,24,\{7,3\}]]_9$ & $[[36,22,\{7,4\}]]_9$ & $[[36,19,\{7,5\}]]_9$\\
15 & $(0,8,4)$ & & $\pmb{[[36,32,\{3,2\}]]_9}$ & $\pmb{[[36,30,\{3,3\}]]_9}$  & $[[36,28,\{4,3\}]]_9$ & $[[36,25,\{5,3\}]]_9$\\
\noalign{\smallskip}
16 & $(1,2,1)$ & & $\pmb{[[37,2,\{32,2\}]]_9}$   &  & &\\ 
17 & $(1,3,2)$ & & $[[37,5,\{27,2\}]]_9$  & $[[37,3,\{27,3\}]]_9$ & $[[37,1,\{27,4\}]]_9$ &\\
18 & $(1,4,3)$ & & $[[37,9,\{22,2\}]]_9$  & $[[37,7,\{22,3\}]]_9$ & $[[37,5,\{22,4\}]]_9$ & $[[37,2,\{22,5\}]]_9$\\
19 & $(1,5,2)$ & & $[[37,12,\{18,2\}]]_9$  & $[[37,10,\{18,3\}]]_9$ & $[[37,8,\{18,4\}]]_9$ & $[[37,5,\{18,5\}]]_9$\\
20 & $(1,5,4)$ & & $[[37,14,\{17,2\}]]_9$  & $[[37,12,\{17,3\}]]_9$ & $[[37,10,\{17,4\}]]_9$ & $[[37,7,\{17,5\}]]_9$\\
\noalign{\smallskip}
21 & $(1,6,3)$ & & $[[37,18,\{13,2\}]]_9$  & $[[37,16,\{13,3\}]]_9$ & $[[37,14,\{13,4\}]]_9$ & $[[37,11,\{13,5\}]]_9$\\
22 & $(1,6,5)$ & & $[[37,20,\{12,2\}]]_9$  & $[[37,18,\{12,3\}]]_9$ & $[[37,16,\{12,4\}]]_9$ & $[[37,13,\{12,5\}]]_9$\\
23 & $(1,7,4)$ & & $[[37,25,\{8,2\}]]_9$  & $[[37,23,\{8,3\}]]_9$  & $[[37,21,\{8,4\}]]_9$ & $[[37,18,\{8,5\}]]_9$\\
24 & $(1,7,6)$ & & $[[37,27,\{7,2\}]]_9$  & $[[37,25,\{7,3\}]]_9$  & $[[37,23,\{7,4\}]]_9$ & $[[37,20,\{7,5\}]]_9$\\
25 & $(1,8,5)$ & & $\pmb{[[37,33,\{3,2\}]]_9}$  & $\pmb{[[37,31,\{3,3\}]]_9}$  & $[[37,29,\{4,3\}]]_9$ & $[[37,26,\{5,3\}]]_9$\\
\noalign{\smallskip}
26 & $(2,2,0)$ & & $\pmb{[[38,1,\{33,2\}]]_9}\star$  &  & &\\
27 & $(2,2,1)$ & & $[[38,2,\{32,2\}]]_9$  &   & &\\
28 & $(2,3,1)$ & & $[[38,4,\{28,2\}]]_9$  & $[[38,2,\{28,3\}]]_9$  & &\\
29 & $(2,3,2)$ & & $[[38,5,\{27,2\}]]_9$  & $[[38,3,\{27,3\}]]_9$  & $[[38,1,\{27,4\}]]_9$ &\\
30 & $(2,4,2)$ & & $[[38,8,\{23,2\}]]_9$  & $[[38,6,\{23,3\}]]_9$  & $[[38,4,\{23,4\}]]_9$ & $[[38,1,\{23,5\}]]_9$\\
\noalign{\smallskip}
31 & $(2,4,3)$ & & $[[38,9,\{22,2\}]]_9$  & $[[38,7,\{22,3\}]]_9$  & $[[38,5,\{22,4\}]]_9$ & $[[38,2,\{22,5\}]]_9$\\
32 & $(2,5,3)$ & & $[[38,13,\{18,2\}]]_9$ & $[[38,11,\{18,3\}]]_9$ & $[[38,9,\{18,4\}]]_9$ & $[[38,6,\{18,5\}]]_9$\\
33 & $(2,5,4)$ & & $[[38,14,\{17,2\}]]_9$  & $[[38,12,\{17,3\}]]_9$  & $[[38,10,\{17,4\}]]_9$ & $[[38,7,\{17,5\}]]_9$\\
34 & $(2,6,4)$ & & $[[38,19,\{13,2\}]]_9$  & $[[38,17,\{13,3\}]]_9$  & $[[38,15,\{13,4\}]]_9$ & $[[38,12,\{13,5\}]]_9$\\
35 & $(2,7,5)$ & & $[[38,26,\{8,2\}]]_9$  & $[[38,24,\{8,3\}]]_9$ & $[[38,22,\{8,4\}]]_9$ & $[[38,19,\{8,5\}]]_9$\\
\noalign{\smallskip}
36 & $(2,8,6)$ & & $\pmb{[[38,34,\{3,2\}]]_9}$  & $\pmb{[[38,32,\{3,3\}]]_9}$ & $[[38,30,\{4,3\}]]_9$ & $[[38,27,\{5,3\}]]_9$\\
37 & $(3,2,0)$ & & $\pmb{[[39,1,\{34,2\}]]_9}\star$  &  & &\\
38 & $(3,2,1)$ & & $\pmb{[[39,2,\{33,2\}]]_9}$  &  & &\\
39 & $(3,3,2)$ & & $[[39,5,\{28,2\}]]_9$  & $[[39,3,\{28,3\}]]_9$  & $[[39,1,\{28,4\}]]_9$ &\\
40 & $(3,4,3)$ & & $[[39,9,\{23,2\}]]_9$  & $[[39,7,\{23,3\}]]_9$  & $[[39,5,\{23,4\}]]_9$ & $[[39,2,\{23,5\}]]_9$\\
\noalign{\smallskip}
41 & $(3,5,4)$ & & $[[39,14,\{18,2\}]]_9$  & $[[39,12,\{18,3\}]]_9$  & $[[39,10,\{18,4\}]]_9$ & $[[39,7,\{18,5\}]]_9$\\
42 & $(3,6,3)$ & & $[[39,18,\{14,2\}]]_9$  & $[[39,16,\{14,3\}]]_9$  & $[[39,14,\{14,4\}]]_9$ & $[[39,11,\{14,5\}]]_9$\\
43 & $(3,6,5)$ & & $[[39,20,\{13,2\}]]_9$  & $[[39,18,\{13,3\}]]_9$  & $[[39,16,\{13,4\}]]_9$ & $[[39,13,\{13,5\}]]_9$\\
44 & $(3,7,6)$ & & $[[39,27,\{8,2\}]]_9$  & $[[39,25,\{8,3\}]]_9$  & $[[39,23,\{8,4\}]]_9$ & $[[39,20,\{8,5\}]]_9$\\
45 & $(3,8,7)$ & & $\pmb{[[39,35,\{3,2\}]]_9}$  & $\pmb{[[39,33,\{3,3\}]]_9}$  & $[[39,31,\{4,3\}]]_9$ & $[[39,28,\{5,3\}]]_9$\\
\noalign{\smallskip}
46 & $(4,2,0)$ & & $\pmb{[[40,1,\{35,2\}]]_9}\star$  &  & &\\
47 & $(4,2,1)$ & & $\pmb{[[40,2,\{34,2\}]]_9}$  &  & &\\
48 & $(4,6,4)$ & & $[[40,19,\{14,2\}]]_9$  & $[[40,17,\{14,3\}]]_9$  & $[[40,15,\{14,4\}]]_9$ & $[[40,12,\{14,5\}]]_9$\\
49 & $(4,6,5)$ & & $[[40,20,\{13,2\}]]_9$  & $[[40,18,\{13,3\}]]_9$  & $[[40,16,\{13,4\}]]_9$ & $[[40,13,\{13,5\}]]_9$\\
50 & $(4,7,3)$ & & $[[40,24,\{10,2\}]]_9$  & $[[40,22,\{10,3\}]]_9$  & $[[40,20,\{10,4\}]]_9$ & $[[40,17,\{10,5\}]]_9$\\
\noalign{\smallskip}
51 & $(5,2,0)$ & & $\pmb{[[41,1,\{36,2\}]]_9}\star$ &  & &\\
52 & $(5,2,1)$ & & $\pmb{[[41,2,\{35,2\}]]_9}$ &  & &\\
\noalign{\smallskip}
\hline
\end{tabular}
\end{table}
\begin{table}
\renewcommand\thetable{5}
\caption{Good $9$-ary AQCs {\it(Continued)}}
\label{tab:9b}
\centering
\setlength{\tabcolsep}{4pt}
\begin{tabular}{c l c l l l l}
\hline
\noalign{\smallskip}
No. & $(t,m,\ell)$ & $(d_{z},\delta)$ & $Q$ from Th.~\ref{th:2} & $Q$ from Th.~\ref{th:3} & $Q$ from Th.~\ref{th:4} & $Q$ from Th.~\ref{th:5} \\
\noalign{\smallskip}
\hline
\noalign{\smallskip}
53 & $(5,3,2)$ & & $[[41,5,\{29,2\}]]_9$  & $[[41,3,\{29,3\}]]_9$ & $[[41,1,\{29,4\}]]_9$ &\\
54 & $(5,4,3)$ & & $[[41,9,\{24,2\}]]_9$  & $[[41,7,\{24,3\}]]_9$ & $[[41,5,\{24,4\}]]_9$ & $[[41,2,\{24,5\}]]_9$\\
55 & $(5,5,3)$ & $(20,19)$ & $[[41,13,\{20,2\}]]_9$ & $[[41,11,\{20,3\}]]_9$ & $[[41,9,\{20,4\}]]_9$ & $[[41,6,\{20,5\}]]_9$\\
56 & $(5,5,4)$ & & $[[41,14,\{19,2\}]]_9$  & $[[41,12,\{19,3\}]]_9$ & $[[41,10,\{19,4\}]]_9$ & $[[41,7,\{19,5\}]]_9$\\
57 & $(5,6,5)$ & & $[[41,20,\{14,2\}]]_9$  & $[[41,18,\{14,3\}]]_9$ & $[[41,16,\{14,4\}]]_9$ & $[[41,13,\{14,5\}]]_9$\\
\noalign{\smallskip}
58 & $(5,7,4)$ & & $[[41,25,\{10,2\}]]_9$  & $[[41,23,\{10,3\}]]_9$ & $[[41,21,\{10,4\}]]_9$ & $[[41,18,\{10,5\}]]_9$\\
59 & $(6,2,0)$ & & $\pmb{[[42,1,\{37,2\}]]_9}\star$   &  & &\\
60 & $(6,2,1)$ & & $\pmb{[[42,2,\{36,2\}]]_9}$   &  &&\\
61 & $(6,3,2)$ & & $[[42,5,\{30,2\}]]_9$  & $[[42,3,\{30,3\}]]_9$ & $[[42,1,\{30,4\}]]_9$ &\\
62 & $(6,5,4)$ & $(20,19)$ & $[[42,14,\{20,2\}]]_9$ & $[[42,12,\{20,3\}]]_9$ & $[[42,10,\{20,4\}]]_9$ & $[[42,7,\{20,5\}]]_9$\\
\noalign{\smallskip}
63 & $(6,7,3)$ & & $[[42,24,\{11,2\}]]_9$  & $[[42,22,\{11,3\}]]_9$ & $[[42,20,\{11,4\}]]_9$ & $[[42,17,\{11,5\}]]_9$\\
64 & $(6,7,5)$ & & $[[42,26,\{10,2\}]]_9$  & $[[42,24,\{10,3\}]]_9$ & $[[42,22,\{10,4\}]]_9$ & $[[42,19,\{10,5\}]]_9$\\
65 & $(7,2,0)$ & & $\pmb{[[43,1,\{38,2\}]]_9}\star$   &  & &\\
66 & $(7,2,1)$ & & $\pmb{[[43,2,\{37,2\}]]_9}$   &  & &\\
67 & $(7,3,1)$ & & $[[43,4,\{32,2\}]]_9$  & $[[43,2,\{32,3\}]]_9$  &&\\
\noalign{\smallskip}
68 & $(7,3,2)$ & & $[[43,5,\{31,2\}]]_9$  & $[[43,3,\{31,3\}]]_9$ & $[[43,1,\{31,4\}]]_9$ &\\
69 & $(7,7,3)$ & $(12,11)$ & $[[43,24,\{12,2\}]]_9$ & $[[43,22,\{12,3\}]]_9$ & $[[43,20,\{12,4\}]]_9$ & $[[43,17,\{12,5\}]]_9$\\
70 & $(7,7,4)$ & & $[[43,25,\{11,2\}]]_9$  & $[[43,23,\{11,3\}]]_9$ & $[[43,21,\{11,4\}]]_9$ & $[[43,18,\{11,5\}]]_9$\\
71 & $(7,7,6)$ & & $[[43,27,\{10,2\}]]_9$  & $[[43,25,\{10,3\}]]_9$ & $[[43,23,\{10,4\}]]_9$ & $[[43,20,\{10,5\}]]_9$\\
72 & $(8,2,0)$ & & $\pmb{[[44,1,\{39,2\}]]_9}\star$   &  & &\\
\noalign{\smallskip}
73 & $(8,2,1)$ & & $\pmb{[[44,2,\{38,2\}]]_9}$   &  & &\\
74 & $(8,3,1)$ & & $[[44,4,\{33,2\}]]_9$  & $[[44,2,\{33,3\}]]_9$ & &\\
75 & $(8,3,2)$ & & $[[44,5,\{32,2\}]]_9$  & $[[44,3,\{32,3\}]]_9$ & $[[44,1,\{32,4\}]]_9$ &\\
76 & $(8,4,2)$ & & $[[44,8,\{27,2\}]]_9$  & $[[44,6,\{27,3\}]]_9$ & $[[44,4,\{27,4\}]]_9$ & $[[44,1,\{27,5\}]]_9$\\
77 & $(8,4,3)$ & & $[[44,9,\{26,2\}]]_9$  & $[[44,7,\{26,3\}]]_9$ & $[[44,5,\{26,4\}]]_9$ & $[[44,2,\{26,5\}]]_9$\\
\noalign{\smallskip}
78 & $(8,7,1)$ & $(14,13)$ & $[[44,22,\{14,2\}]]_9$ & $[[44,20,\{14,3\}]]_9$ & $[[44,18,\{14,4\}]]_9$ & $[[44,15,\{14,5\}]]_9$\\
79 & $(8,7,5)$ & & $[[44,26,\{11,2\}]]_9$  & $[[44,24,\{11,3\}]]_9$ & $[[44,22,\{11,4\}]]_9$ & $[[44,19,\{11,5\}]]_9$\\
80 & $(9,2,0)$ & & $\pmb{[[45,1,\{40,2\}]]_9}\star$   &  & &\\
81 & $(9,2,1)$ & & $\pmb{[[45,2,\{39,2\}]]_9}$   &  &&\\
82 & $(9,3,1)$ & & $[[45,4,\{34,2\}]]_9$  & $[[45,2,\{34,3\}]]_9$ & &\\
\noalign{\smallskip}
83 & $(9,3,2)$ & & $[[45,5,\{33,2\}]]_9$  & $[[45,3,\{33,3\}]]_9$ & $[[45,1,\{33,4\}]]_9$ & \\
84 & $(9,4,2)$ & & $[[45,8,\{28,2\}]]_9$  & $[[45,6,\{28,3\}]]_9$ & $[[45,4,\{28,4\}]]_9$ & $[[45,1,\{28,5\}]]_9$\\
85 & $(9,4,3)$ & & $[[45,9,\{27,2\}]]_9$  & $[[45,7,\{27,3\}]]_9$ & $[[45,5,\{27,4\}]]_9$ & $[[45,2,\{27,5\}]]_9$\\
86 & $(9,7,0)$ & $(15,14)$ & $[[45,21,\{15,2\}]]_9$ & $[[45,19,\{15,3\}]]_9$ & $[[45,17,\{15,4\}]]_9$ & $[[45,14,\{15,5\}]]_9$\\
87 & $(9,7,6)$ & & $[[45,27,\{11,2\}]]_9$  & $[[45,25,\{11,3\}]]_9$ & $[[45,23,\{11,4\}]]_9$ & $[[45,20,\{11,5\}]]_9$\\
\noalign{\smallskip}
\hline
\end{tabular}
\end{table}

\section{Conclusion}
Nested XL codes provide good ingredients to derive pure $q$-ary CSS AQCs. In many cases the derived quantum codes can be 
shown to be optimal or best-known. For $d_{x}\geq 4$ the derived AQCs have very little, if at all, overlap with previously 
known ones. The comparison of parameters are done with respect to the list provided in~\cite{Gua13} and the more recent 
results in~\cite{EJLP12}.

In the cases where $d_{x} \in \{4,5\}$, subjecting the pair $C_{1}^\perp \subset C_{2}$ based on the XL codes treated in this paper to 
the so-called triangle bound in~\cite[Section V]{EJLP12} may certify that the derived good AQCs are in fact best possible or best known. 
More generally, for lengths beyond those covered by the XL codes, one can perhaps look at families of (nested) polynomial 
codes to come up with good AQCs.
\section*{Appendix A: Proof of Linear Independence of the Column Vectors in (\ref{M3})}
\begin{proof}
We consider all possible cases separately.

\noindent{Case I.} When $a^{q+1} = b^{q+1} = c^{q+1}$, Lemma~\ref{a0} says that $a^q+a, b^q+b, c^q+c$ 
are pairwise distinct. Hence, the matrix
\[\left(\begin{array}{ccc}
1&1&1\\
a^q+a	    &b^q+b       &c^q+c\\
(a^q+a)^2	&(b^q+b)^2   &(c^q+c)^2  
\end{array} \right)
\]
is a Vandermonde matrix, establishing linear independence for this case.

\noindent{Case II.} When $a^{q+1}= b^{q+1} \neq c^{q+1}$, using $b^q+b \neq a^q+a$ from Lemma~\ref{a0}, verify that
\begin{align*}
\left|\begin{array}{ccc}
1&1&1\\
a^{q+1}	    &b^{q+1}      &c^{q+1}\\
a^q+a	&b^q+b   &c^q+c
\end{array} \right|
&=
\left|\begin{array}{ccc}
1&1&1\\
0	    &0       &c^{q+1}-a^{q+1}\\
0  	    & (b^q+b)- (a^q+a)  & (c^q+c)- (a^q+a) 
\end{array} \right|\\
&= (c^{q+1}-a^{q+1})[(a^q+a)-(b^q+b)]\neq 0 \text{.}
\end{align*}
This case is thus settled.

\noindent{Case III.} The elements $a^{q+1}, b^{q+1}$, and $c^{q+1}$ are pairwise distinct. Then 
\[\left(\begin{array}{ccc}
1&1&1\\
a^{q+1}	    &b^{q+1}      &c^{q+1}\\
a^{2(q+1)}	&b^{2(q+1)}      &c^{2(q+1)} 
\end{array} \right)
\]
is a Vandermonde matrix. Hence, the columns in (\ref{M3}) are linearly independent. 
The proof is now complete.
\end{proof}

\section*{Appendix B: Proof of Linear Independence of the Column Vectors in (\ref{M4})}
\begin{proof} There are four cases to consider.
	
\noindent{Case I.} When $ a^{q+1} = b^{q+1} = c^{q+1} = d^{q+1}$, Lemma~\ref{a0} says that 
$a^q+a, b^q+b, c^q+c$ and $d^q+d$ are pairwise distinct. The following Vandermonde matrix certifies linear independence
\[\left(\begin{array}{cccc}
1&1&1&1\\
a^q+a	    & b^q+b       & c^q+c     & d^q+d\\
(a^q+a)^2	& (b^q+b)^2   & (c^q+c)^2 & (d^q+d)^2\\
(a^q+a)^3	& (b^q+b)^3   & (c^q+c)^3 & (d^q+d)^3
\end{array} \right)\text{.}
\]

\noindent{Case II.} Let there be two distinct elements among $a^{q+1}, b^{q+1}, c^{q+1}$, and $d^{q+1}$.
Lemma~\ref{a0} says that if $ a^{q+1} = b^{q+1} = c^{q+1} \neq d^{q+1}$, then $ a^q+a, b^q+b$, and $c^q+c$ 
are pairwise distinct. Let $D_{\M_{1}}$ be the determinant of the matrix
\[
\M_{1} = 
\left(\begin{array}{cccc}
1&1&1&1\\ 
a^q+a	    & b^q+b       & c^q+c      & d^{q}+d    \\
(a^q+a)^2	& (b^q+b)^2   & (c^q+c)^2  & (d^{q}+d)^{2} \\
a^{q+1}     & b^{q+1}     & c^{q+1}    & d^{q+1} \\
\end{array} 
\right)\text{.}
\]
Denote by $\cN$ the Vandermonde matrix with nonzero determinant $D_{\cN}$ derived by deleting the last 
column and the last row of $\M_{1}$. Subtracting $c^{q+1}$ times the first row from the last row of $\M_{1}$ 
reveals that
\[
D_{\M_{1}}=(d^{q+1}-c^{q+1}) D_{\cN} \neq 0 \text{.}
\]

When $a^{q+1} = b^{q+1} \neq c^{q+1} =d^{q+1}$, then, by Lemma \ref{a0}, 
$a^q+a	\neq b^q+b$ and $c^q+c \neq d^q+d$. For brevity, let $\lambda_{y}=y^{2q+1}+y^{q+2}$ 
for $y \in \{a,b,c,d\}$. Let $D_{\M_{2}}$ be the determinant of
\[
\M_{2} = 
\left(\begin{array}{cccc}
1&1&1&1\\ 
a^{q+1}	         & b^{q+1}          & c^{q+1}          & d^{q+1}    \\
a^{q}+a          & b^{q}+b          & c^{q}+c          & d^{q}+d    \\
\lambda_{a}      & \lambda_{b}      & \lambda_{c}      & \lambda_{d} \\
\end{array} 
\right)\text{.}
\]
First, subtract $c^{q+1}$ times the first row from the second row of $\M_{2}$. Then subtract $d^{q+1}$ times 
the first row from the third row of the resulting matrix. Subtract $d^{q+1}$ times the third row from the fourth row of 
this last matrix to get a matrix which we call $\mathcal{P}$. Using the cofactor expansion along 
the fourth column of $\mathcal{P}$ tells us that
\begin{equation*}
D_{\M_{2}}=[(c^{q}+c)-(d^{q}+d)][(b^{q}+b)-(a^{q}+a)]\cdot(a^{q+1}-c^{q+1})(a^{q+1}-d^{q+1}) \neq 0 \text{.}
\end{equation*}

\noindent{Case III.} Without lost of generality, let us assume that $a^{q+1}, b^{q+1}$, 
and $c^{q+1}$ are pairwise distinct and $c^{q+1}= d^{q+1}$. 
Let $D_{\M_{3}}$ be the determinant of
\[
\M_{3} = 
\left(
\begin{array}{cccc}
1&1&1&1\\ 
a^{q+1}	    & b^{q+1}       & c^{q+1}      & d^{q+1}    \\
a^{2(q+1)}	& b^{2(q+1)}    & c^{2(q+1)}   & d^{2(q+1)} \\
a^{q}+a     & b^{q}+b       & c^{q}+c      & d^{q}+d    \\
\end{array} 
\right)\text{.}
\]
Also, let $\cN'$ be the Vandermonde matrix with nonzero determinant $D_{\cN'}$ derived by deleting the last 
column and the last row of $\M_{3}$. Subtracting the third column from the fourth column of $\M_{3}$ and 
using the cofactor expansion along the fourth row of the resulting matrix shows that
\[
D_{\M_{3}}=[(d^{q}+d)-(c^{q}+c)] D_{\cN'} \neq 0 \text{.}
\]
We conclude that, in this case, the columns in (\ref{M4}) are linearly independent.

\noindent{Case IV.} In the case where $a^{q+1}, b^{q+1}, c^{q+1}$ and $d^{q+1}$ are pairwise distinct, 
we have a Vandermonde matrix
\[
\left(\begin{array}{cccc}
	1&1&1&1\\ 
	a^q+a	    &b^q+b       &c^q+c    &d^q+d\\
	(a^q+a)^2	&(b^q+b)^2   &(c^q+c)^2&(d^q+d)^2\\
	(a^q+a)^3	&(b^q+b)^3   &(c^q+c)^3&(d^q+d)^3 
\end{array}\right)\text{.}
\]
The columns in (\ref{M4}) are, therefore, linearly independent.
\end{proof}
\begin{acknowledgements}
We thank Markus Grassl and Dimitrii Pasechnik for some computer algebra pointers, 
and the anonymous referees for their comments and suggestions.
\end{acknowledgements}


\end{document}